\documentclass[11pt,lettersize]{article}

\usepackage{authblk}
\usepackage{fullpage}
\usepackage{amsmath,amsthm,amssymb}
\usepackage{todonotes}
\presetkeys{todonotes}{inline}{}
\usepackage{xspace} 
\usepackage{complexity} 
\usepackage{thm-restate}

\usepackage{cleveref} 

\usepackage{algpseudocode}
\usepackage{algorithm}

\usepackage{tcolorbox}

\usepackage{tikz}
	\usetikzlibrary{calc,fit,arrows.meta}
    \usetikzlibrary{decorations.pathreplacing} 

\usepackage{xcolor}

\definecolor{myred}{HTML}{db3f3d}
\definecolor{myblue}{HTML}{0072BD}

\newtheorem{theorem}{Theorem}[section]
\newtheorem{proposition}[theorem]{Proposition}

\newtheorem{lemma}[theorem]{Lemma}
\newtheorem{corollary}[theorem]{Corollary}

\newtheorem{claim}[theorem]{Claim}

\theoremstyle{definition}

\newtheorem{example}[theorem]{Example}

\usepackage[round]{natbib}
\title{Robust Popular Matchings}
\date{}

\author[1]{Martin Bullinger}
\author[2]{Gergely Csáji}
\author[3]{Rohith Reddy Gangam}
\author[3]{Parnian Shahkar}

\affil[1]{ \small School of Engineering Mathematics and Technology, University of Bristol, UK}
\affil[2]{\small ELTE KRTK KTI, Hungary}
\affil[3]{ \small Department of Computer Science, University of California Irvine, USA\protect\\ \vspace*{0.05cm} martin.bullinger@bristol.ac.uk, csaji.gergely@krtk.elte.hu, rgangam@uci.edu, shahkarp@uci.edu}

\newtcolorbox{wbox}
{
	colback  = white,
}

\newcommand{\MatP}{MP\xspace} 
\newcommand{\ins}{\mathcal I} 
\newcommand{\iA}{\mathcal I_A} 
\newcommand{\iB}{\mathcal I_B} 
\newcommand{\iX}{\mathcal I_X} 
\newcommand{\iH}{\mathcal H_e} 
\newcommand{\iHH}{\mathcal H_{e,e'}} 
\newcommand{\iP}{(\iA,\iB)} 

\newcommand{\vote}{\text{vote}} 
\newcommand{\pmarg}{\Delta} 
\newcommand{\UI}{\mathcal U^{\ins}} 
\newcommand{\UA}{\mathcal U^{\iA}} 
 
\newcommand{\gadg}{\mathcal{G}}

\newcommand{\RobustProb}{\textsc{RobustPopularMatching}\xspace}
\newcommand{\RobustDom}{\textsc{RobustDominantMatching}\xspace}
\newcommand{\ForbEdge}{\textsc{2-ForbiddenEdge}\xspace}
\newcommand{\ForbEdgeForceVert}{\textsc{ForbiddenEdgeForceVert}\xspace}
\newcommand{\PopEdge}{\textsc{PopularEdge}\xspace}
\newcommand{\DomEdge}{\textsc{DominantEdge}\xspace}
\newcommand{\StabMat}{\textsc{StableMatching}\xspace}
\newcommand{\mSAT}{\textsc{monotone-3-SAT}}
\newcommand{\tSAT}{\textsc{3-SAT}}
\newcommand{\nega}{\overline{a}}
\newcommand{\negb}{\overline{b}}
\newcommand{\negc}{\overline{c}}
\newcommand{\negd}{\overline{d}}
\newcommand{\nege}{\overline{e}}
\newcommand{\negf}{\overline{f}}

\begin{document}

\maketitle 

\begin{abstract}
	We study popularity for matchings under preferences.
    This solution concept captures matchings that do not lose against any other matching in a majority vote by the agents.
    A popular matching is said to be \emph{robust} if it is popular among multiple instances.
    We present a polynomial-time algorithm for deciding whether there exists a robust popular matching if instances only differ with respect to the preferences of a single agent.
    The same method applies also to dominant matchings, a subclass of maximum-size popular matchings. 
    By contrast, we obtain \NP-completeness if two instances differ only by two agents of the same side or by a swap of two adjacent alternatives by two agents.
    The first hardness result applies to dominant matchings as well.
    Moreover, we find another complexity dichotomy based on preference completeness for the case where instances differ by making some options unavailable.
    We conclude by discussing related models, such as strong and mixed popularity.
\end{abstract}

\section{Introduction}

Matching under preferences has been a long-standing subject of study for several decades with diverse applications spanning labor markets, organ transplantation, and dating.
The general idea is to match two types of agents that each possess a ranking of the agents from the other side.
One of the most celebrated results in this area is the Deferred Acceptance Algorithm by \citet{GaSh62a} for identifying so-called stable matchings.
These are matchings that do not admit a blocking pair of agents preferring each other to their designated matching partners.
Subsequently, many related algorithms and solution concepts have been developed and investigated.
Among these, the concept of popular matchings proposed by \citet{Gard75a} has caused substantial research, see, e.g., the book chapter by \citet{Cseh17a}.
A matching is said to be popular if it does not lose a majority election against any other matching.
In this election, the agents vote according to their preferences between their respective matching partners.
As already shown by \citeauthor{Gard75a}, stable matchings are popular, but the converse is not necessarily true.

A common feature of real-world scenarios is that it can be difficult for agents to express their exact preferences.
For instance, an agent might report their preferences but revise them at a later stage.
In a matching market, this situation can easily arise when interactions with other agents influence opinions about them.
Alternatively, an agent might maintain their preferences, but some event may occur that renders certain options unacceptable or impossible, or introduces new opportunities.
Again, such situations occur frequently, for example, when some potential matching partners move away, or when an agent meets new ones.

In terms of algorithmic solutions, it is desirable to establish a solution that is robust to changes.
To formalize this idea, we propose \emph{robust popular matchings}, which are popular matchings across multiple instances.
We then consider \RobustProb, the algorithmic problem of computing a robust popular matching, or deciding that no such matching exists.
Specifically, we consider this problem in the two scenarios described above.
First, we assume that the set of available matching partners remains fixed, but agents may \emph{alter their preferences}.
We present a polynomial-time algorithm for the case where only a single agent alters their preferences.
The key idea behind this algorithm is to define a set of hybrid instances, on which we search for popular matchings that include a predefined edge.
This algorithm extends to dominant matchings, which are an important subclass of maximum-size popular matchings \citep{CsKa17a}.
Hence, we can also solve the analogous \RobustDom problem if a single agent alters their preferences.
This result is tight with respect to the number of agents. 
Indeed, we show that \RobustProb and \RobustDom become \NP-complete when two agents from the same side alter their preferences.
For \RobustProb, we obtain a second \NP-completeness for the case when two agents are allowed to perform a swap of two adjacent alternatives, which is arguably the smallest possible preference alteration.
Second, we consider \RobustProb in the setting where agents' preference orders are fixed, but some options may \emph{become unavailable}.
We identify a complexity dichotomy based on whether one of the input instances includes all potential partners as available.
We conclude by discussing related problems, such as robustness with respect to other notions of popularity.

\section{Related Work}
\label{sec:relatedwork}

Popularity was first considered by \citet{Gard75a} under the name of ``majority assignments.''
He also introduced strong popularity, a stricter version of the concept in which a matching must defeat every other matching in a majority vote. 
Hence, unlike standard popularity, ties against other matchings are not sufficient for a matching to be strongly popular.
In a broader interpretation, popular and strongly popular matchings correspond to weak and strong Condorcet winners as studied in social choice theory \citep{Cond85a}.
The book chapters by \citet{Cseh17a} and \citet[Chapter~7]{Manl13a} provide an excellent introduction to popular matchings.

Our research continues a stream of algorithmic results on popularity.
In this line of work, close relationships between popularity and stability often play an important role:
Popular matchings can have different sizes and stable matchings are popular matchings of minimum size.
An important subclass of popular matchings are dominant matchings, which are popular matchings defeating every larger matching in a popularity vote \citep{CsKa17a}.
These matchings always exist and can be computed efficiently, which implies that some maximum-size popular matching can be computed efficiently \citep{CsKa17a}.
By contrast, somewhat surprisingly, it is \NP-hard to decide if there exists a popular matching that is neither stable nor dominant \citep{FKPZ19a}.

An important algorithmic problem for our research is the problem of computing a popular matching containing a predefined set of edges.
A polynomial time algorithm for this problem exists if only a single edge has to be included in the matching \citep{CsKa17a} but the problem is \NP-hard if at least two edges are forced \citep{FKPZ19a}.
However, this hardness heavily relies on the fact that some matching partners are unavailable.
If the preference orders of the agents encompass the complete set of agents of the other side, then popular matchings of maximum weight can be computed in polynomial time \citep{CsKa17a}.
Consequently, by setting appropriate weights, one can find popular matchings containing or preventing any subset of edges for complete instances.

Popular matchings have also been considered in related domains.
\citet{BIM10a} consider popular matchings for weak preferences, where their computation becomes \NP-hard.
However, in a house allocation setting, where one side of the agents corresponds to objects without preferences, popular matchings can be computed efficiently, even for weak preferences \citep{AIKM07a}.
One can also relax the allowed input instances by considering the roommate setting where every pair of agents may be matched.
Then, popularity already leads to computational hardness for strict preferences \citep{GMSZ19a,FKPZ19a}.

In addition, some work considers a randomized variant of popularity called mixed popularity, where the output is a probability distribution over matchings \citep{KMN11a}.
Mixed popular matchings are guaranteed to exist by the Minimax Theorem, and can be computed efficiently if the output are matchings, even in the roommate setting \citep{KMN11a,BrBu20a}.
However, finding matchings in the support of popular outcomes becomes intractable in a coalition formation scenario where the output may contain coalitions of size larger than two.
Once coalitions of size three are allowed, deciding whether a popular partition exists is known to be \coNP-hard and their verification is \coNP-complete \citep{BrBu20a}.
These results suggest that the exact complexity in this setting may be $\Sigma_2^p$-completeness.
In fact, $\Sigma_2^p$-completeness of popularity holds in additively separable and fractional hedonic games, two prominent models of coalition formation that can be seen as an extension of matchings under preferences \citep{BuGi25a}.

A series of papers study the complexity of finding stable matchings across multiple instances.
\citet{MaVa18a} initiate this stream of work and propose a polynomial-time algorithm if one agent shifts down a single alternative.
This result was improved by \citet{GMRV22a} for arbitrary changes of the preference order by a single agent and subsequently even for arbitrary changes by all agents of one side \citep{GMRV23b}.
Notably, the algorithmic approach of these papers is to exploit the combinatorial structure of the lattice of stable matchings.
By contrast, a similar structure for popular matchings is unknown, and we develop an alternative technique for our positive result.
Moreover, a positive result for popularity when multiple agents from the same side alter their preferences is unlikely because we obtain \NP=hardness even for preference changes by two agents from the same side.
In addition, if all agents may change their preference lists, a computational intractability is also obtained for stable matchings \citep{MiOk19a}.
Notably, in the reduction by \citet{MiOk19a}, the number of agents that change their preference order is linear with respect to the total number of agents, and these agents apply extensive changes, which again is a contrast to our hardness results.

Furthermore, popular matchings across multiple instances have been considered in a concurrent paper by \citet{Csaj24a}.
His multilayer model corresponds to our model of perturbed instances.
In addition, he studies popularity among succinctly encoded, possibly exponentially large sets of profiles.
First, he studies a model in which popularity must hold for every preference profile when each agent can have any one of at most a fixed number of preference orders.
Second, he introduces a different concept of robustness, in which popularity is required to hold for all profiles within a small swap distance.
Generally, he finds hardness results for two-sided preferences and polynomial-time solvability for the house allocation setting.
In particular, the results by \citet{Csaj24a} for his multilayer setting imply \NP-completeness of our notion of robust popularity when only one side of the agents perturbs their preferences, which is already a contrast to the polynomial-time algorithm by \citet{GMRV23b} for robust stability.
Furthermore, \citet{Csaj24a} obtains an \NP-completeness result for robust dominant matchings for the case where any number of agents can perturb their preferences.
Notably, this hardness result relies on changing the preferences of an unbounded number of agents.\footnote{The reduction is from the problem of computing a stable matching across two instances by \citet{MiOk19a}.
In the reduction by \citet{Csaj24a}, the number of agents with different preferences in the reduced instance is the same as the number of agents in the source instance.
As discussed before, the number of agents with perturbed preferences in the reduction by \citet{MiOk19a} is linear in the number of agents, and this property is inherited by the reduction of \citet{Csaj24a}.} 
We improve both hardness results by showing them for the case when only two agents from the same side change their preferences.

Additionally, some work on stable matchings considers other models in which multiple instances interact.
\citet{ABG+20a,ABF+22a} propose a model with uncertainty for the true preference relations.
They ask for matchings that are possibly or necessarily stable, or stable with high probability.
In addition, \citet{BBHN21a} study bribery for stable matchings, i.e., the strategic behaviour for achieving certain goals like forcing a given edge into a stable matching by deliberately manipulating preference orders.

Finally, robustness of outcomes with respect to perturbations of the input has also been studied in other scenarios of multiagent systems, such as voting \citep{FaRo15a,SYE13a,BFK+21a}.
There, robustness is commonly studied under the lens of bribery, i.e., deliberately influencing an election by changing its input.
Then, a cost is incurred for modifying votes, often measured with respect to the swap distance of the original and modified votes \citep{EFS09a}.

\section{Preliminaries}
In this section, we introduce our formal model as well as other important computational problems.

\subsection{Popular Matchings}
An instance $\ins$ of \emph{matchings under preferences} (\MatP) consists of a bipartite graph $G^{\ins} = (W\cup F,E^{\ins})$, where the vertices $W$ and $F$ are interpreted as sets of \emph{workers} and \emph{firms}, respectively.
Elements of the union $W\cup F$ are referred to as \emph{agents}.
Each agent $x\in W\cup F$ is equipped with a linear order $\succ_x^{\ins}$, their so-called \emph{preference order}, over $N_x^{\ins} := \{y\in W\cup F\colon \{x,y\}\in E^{\ins}\}$, that is, the set of their neighbors in $G$.
Note that most of our notation uses superscripts to indicate the instance, as we will soon consider multiple instances in parallel.
However, we may omit the superscript when the instance is clear from context.
Moreover, since the sets of workers and firms are identical across instances, we omit superscripts for them entirely.

Given a graph $G = (W\cup F,E)$, a \emph{matching} is a subset $M\subseteq E$ of pairwise disjoint edges, i.e., $m\cap m' = \emptyset$ for all $m, m'\in M$.
For a matching $M$, we call an agent $x\in W\cup F$ \emph{matched} if there exists $m\in M$ with $x\in m$, and \emph{unmatched}, otherwise.
If $x$ is matched, we denote their matching partner by $M(x)$.

Assume now that we are given an instance $\ins$ of \MatP together with an agent $x\in W\cup F$ and two matchings $M$ and $M'$.
We say that $x$ \emph{prefers} $M$ over $M'$ if $x$ is matched in $M$ and unmatched in $M'$, or if $x$ is matched in both $M$ and $M'$ and $M(x) \succ_x M'(x)$.

The notion of popularity depends on a majority vote of the agents between matchings according to their preferences.
Therefore, we define the following notation for a vote between matchings: 
\begin{align*}
    \vote^{\ins}_x(M,M') := \begin{cases}
        1 & x \text{ prefers } M \text{ over }M'\text,\\
        -1 & x \text{ prefers } M' \text{ over }M\text,\\
        0 & \text{otherwise.}\\
    \end{cases}
\end{align*}

Given a set of agents $N\subseteq W\cup F$, we define $\vote_N^{\ins}(M,M') := \sum_{x\in N} \vote^{\ins}_x(M,M')$.
The \emph{popularity margin} between $M$ and $M'$ is defined as $\pmarg^{\ins}(M,M') := \vote_{W\cup F}^{\ins}(M,M')$.
Now, a matching $M$ is said to be \emph{popular} with respect to instance $\ins$ if, for every matching~$M'$, it holds that $\pmarg^{\ins}(M,M')\ge 0$.
In other words, a matching is popular if it does not lose a majority vote among the agents in an election against any other matching.
Moreover, a matching $M$ is said to be \emph{stable} if for every edge $e = \{x,y\}\in E\setminus M$, it holds that $x$ is matched and prefers $M(x)$ to $y$ or $y$ is matched and prefers $M(y)$ to~$x$.
As discussed before, stable matchings are popular.
Finally, a matching $M$ is said to be \emph{dominant} if it is popular and if, for every matching $M'$ with $|M'| > |M|$, it holds that $\pmarg^{\ins}(M,M') > 0$.
Hence, a dominant matching is a maximum-size popular matching but the converse is not true \citep{CsKa17a}.

Finally, we state a characterization for popular matchings that was proved by
\citet[Theorem~1]{HuKa11a} and will be useful for verifying the popularity of matchings.

Given a matching $M$, we call an edge $e$ a $(-,-)$ edge if both endpoints prefer $e$ to $M$, a $(-,+)$-edge if exactly one endpoint prefers $e$ to $M$, a $(+,+)$ edge if neither endpoint prefers $e$ to $M$, and a $(0,0)$ edge if $e\in M$.
Let $G_M$ be the subgraph of the $(-,-)$ and $(+,-)$ and $(0,0)$ edges, i.e., the subgraph after deletion of the $(+,+)$ edges.

\begin{theorem}[\citet{HuKa11a}]\label{thm:pop_char}
    A matching $M$ is popular if and only if the following three conditions are satisfied in the subgraph $G_M$:
\begin{align*}
\text{(i)} &\quad \text{There is no alternating cycle with respect to } M \text{ that contains a } (-,-) \text{ edge.} \\
\text{(ii)} &\quad \text{There is no alternating path with respect to $M$ starting from an unmatched vertex of } M \\
           &\quad \text{that contains a } (-,-) \text{ edge.} \\
\text{(iii)} &\quad \text{There is no alternating path with respect to $M$ ending in } (-,-) \text{ edges.}
\end{align*}
\end{theorem}

\subsection{Robust Popularity}
We are interested in matchings that are popular across multiple instances.
For this, we consider a pair of instances $\iP$ of {\MatP} where we assume that they are defined for the same set of workers and firms.

Given such a pair $\iP$, a matching is called a \emph{robust popular matching} with respect to $\iA$ and $\iB$ if it is popular with respect to both $\iA$ and $\iB$ individually.
Note that this implies that a robust popular matching is, in particular, a matching for both $\iA$ and $\iB$ and, therefore, a subset of the edge set of both underlying graphs.
We are interested in the computational problem of computing robust popular matchings, conceptualized as follows.

	\begin{wbox}
		\RobustProb\\ 
		\textbf{Input:} Pair $\iP$ of instances of \MatP.\\ 
		\textbf{Question:} Does there exist a robust popular matching with respect to $\iA$ and $\iB$?
	\end{wbox}

Similarly, a matching is called a \emph{robust dominant matching} with respect to $\iA$ and $\iB$ if it is dominant with respect to both $\iA$ and $\iB$ individually.

	\begin{wbox}
		\RobustDom\\
		\textbf{Input:} Pair $\iP$ of instances of \MatP.\\ 
		\textbf{Question:} Does there exist a robust dominant matching with respect to $\iA$ and $\iB$?
	\end{wbox}

Finally, a maximum-size robust popular matching is a robust popular matching that has maximum cardinality among robust popular matchings.
Clearly, since dominant matchings are maximum-size popular matchings, a robust dominant matching is also a maximum-size robust popular matching.
However, as we show in \Cref{app:dominant}, the converse is not true: there exist instances that do not admit a robust dominant matching, but they contain a robust popular matching and, therefore, a maximum-size robust popular matching.

We study our algorithmic problems for the cases of input pairs, where the underlying graph or the underlying preferences remain the same across instances.

First, if the underlying graphs are the same, i.e., $G^{\iA} = G^{\iB}$, we say that $\iB$ is a \emph{perturbed instance} with respect to $\iA$. 
Hence, a perturbed instance only differs with respect to the preference orders of the agents over the identical sets of neighbors.
As a special case, we consider the case where agents simply swap two adjacent alternatives in their preference order.
Given an agent $x\in W\cup F$ and two preference orders $\succ_x^{\iA}, \succ_x^{\iB}$ over $N_x$, we say that $\succ_x^{\iB}$ evolves from $\succ_x^{\iA}$ by a \emph{swap} if there exists a pair of agents $y,y'\in N_x$ such that
\begin{itemize}
    \item $y \succ_x^{\iA} y'$,
    \item $y' \succ_x^{\iB} y$, and
    \item for all pairs of agents $z, z' \in N_x$ with $\{z,z'\}\neq \{y,y'\}$, it holds that $z \succ_x^{\iA} z'$ if and only if $z \succ_x^{\iB} z'$.
\end{itemize}
Note that the third properties implies that $y$ and $y'$ have to be adjacent in the preference order.
We say that $\iB$ evolves from $\iA$ by \emph{swaps} if, for all agents $x\in W\cup F$ with $\succ_x^{\iA} \neq \succ_x^{\iB}$, it holds that $\succ_x^{\iB}$ evolves from $\succ_x^{\iA}$ by a swap.

Second, we consider the case of identical preference orders.
More formally, we say that $\iB$ evolves from $\iA$ by \emph{altering availability} if, for every agent $x\in W\cup F$, there exists a preference order $\succ_x$ on $N_x^{\iA}\cup N_x^{\iB}$ such that for all $y,z\in N_x^{\iA}$, it holds that $y\succ_x^{\iA} z$ if and only if $y\succ_x z$ and for all $y,z\in N_x^{\iB}$, it holds that $y\succ_x^{\iB} z$ if and only if $y\succ_x z$.
In other words, the underlying graphs of the two input instances may differ but the preferences for common neighbors are identical.
As a special case, we say that $\iB$ evolves from $\iA$ by \emph{reducing availability} if $\iB$ evolves from $\iA$ by altering availability and $E^{\iB}\subseteq E^{\iA}$.

Before presenting our results, we illustrate central concepts in an example.

\begin{example}\label{ex:basic}
    Consider an instance $\iA$ of {\MatP} with $W = \{w_1, w_2, w_3, w_4\}$ and $F = \{f_1, f_2, f_3, f_4\}$ where the graph and preferences are defined as in \Cref{fig:example}.
    \begin{figure}
        \centering
        \begin{tikzpicture}
            \pgfmathsetmacro\graphspan{2.5}
            \pgfmathsetmacro\graphheight{1.2}
            \pgfmathsetmacro\nodesize{.6cm}
            \node[draw, circle, minimum size = \nodesize, label = 180:$f_2\succ^{\iA}_{w_1} f_1\succ^{\iA}_{w_1} f_3$] (w1) at (0,\graphheight) {$w_1$};
            \node[draw, circle, minimum size = \nodesize, label = 180:$f_1\succ^{\iA}_{w_2} f_3$] (w2) at (0,0) {$w_2$};
            \node[draw, circle, minimum size = \nodesize, label = 180:$f_2\succ^{\iA}_{w_3} f_1\succ^{\iA}_{w_3} f_4$] (w3) at (0,-\graphheight) {$w_3$};
            \node[draw, circle, minimum size = \nodesize, label = 180:$f_2\succ^{\iA}_{w_4} f_4$] (w4) at (0,-2*\graphheight) {$w_4$};

            \node[draw, circle, minimum size = \nodesize, label = 0:$w_3\succ^{\iA}_{f_1} w_1\succ^{\iA}_{f_1} w_2$] (f1) at (\graphspan,\graphheight) {$f_1$};
            \node[draw, circle, minimum size = \nodesize, label = 0:$w_3\succ^{\iA}_{f_2} w_1\succ^{\iA}_{f_2} w_4$] (f2) at (\graphspan,0) {$f_2$};            \node[draw, circle, minimum size = \nodesize, label=0:$w_1\succ^{\iA}_{f_3} w_2$] (f3) at (\graphspan,-\graphheight) {$f_3$};
            \node[draw, circle, minimum size = \nodesize, label=0:$w_4\succ^{\iA}_{f_4} w_3$] (f4) at (\graphspan,-2*\graphheight) {$f_4$};
            
            \draw[thick] (w3) -- (f4) -- (w4) -- (f2) -- (w3) -- (f1) -- (w2) -- (f3) -- (w1) -- (f2);
            \draw[thick] (w1) -- (f1);
        \end{tikzpicture}
        \caption{Instance $\iA$ in \Cref{ex:basic}. The perturbed instance~$\iB$ is obtained by having agent $w_1$ swap their preferences for $f_1$ and $f_3$.
        }
        \label{fig:example}
    \end{figure}
    The instance contains a unique stable matching $M_1 = \{\{w_1,f_1\}, \{w_2,f_3\}, \{w_3,f_2\}, \{w_4,f_4\}\}$.
    Moreover, there exists another popular matching $M_2 = \{\{w_1,f_3\}, \{w_2,f_1\}, \{w_3,f_2\}, \{w_4,f_4\}\}$. 
    Note that this matching is not stable, because of the edge $\{w_1,f_1\}$.
    Additionally, since $M_1$ and $M_2$ match all agents, both of them are dominant matchings.

    Now, consider the instance $\iB$ that is obtained from instance $\iA$ by having agent~$w_1$ change their preferences to $f_2\succ^{\iB}_{w_1} f_3\succ^{\iB}_{w_1} f_1$, and leaving everything else the same.
    Hence, $\iB$ evolves from $\iA$ by a swap in the preferences of agent~$w_1$.
    The unique popular (and, therefore, stable) matching in $\iB$ is $M_2$.
    Therefore, $\iP$ is a Yes-instance of \RobustProb.
    Moreover, since $M_2$ still matches all agents in $\iB$, it is also a robust dominant matching for $\iB$.
    Hence, $\iP$ is also a Yes-instance of \RobustDom. \hfill$\lhd$
\end{example}

\subsection{Related Algorithmic Problems}\label{sec:problems}

Our results make use of other algorithmic problems considered in the literature.
First, we consider \StabMat, the problem of computing a stable matching in a given instance of {\MatP}, which can be solved in polynomial time by the famous Deferred Acceptance Algorithm \citep{GaSh62a}.

	\begin{wbox}
		\StabMat\\
		\textbf{Input:} Instance $\ins$ of \MatP.\\
		\textbf{Task:} Compute a stable matching in $\ins$.
	\end{wbox}

Second, we consider \PopEdge, the problem of computing a popular matching in a given instance of {\MatP} containing a designated edge, or deciding that no such matching exists.
This problem can also be solved in polynomial time \citep{CsKa17a}.

	\begin{wbox}
		\PopEdge\\
		\textbf{Input:} Instance $\ins$ of \MatP and designated edge $e\in E^{\ins}$.\\
		\textbf{Task:} Compute a popular matching in $\ins$ containing $e$ or determine that there exists no such matching.
	\end{wbox}

Third, we consider \DomEdge, the analogous problem for dominant matchings.
Once again, this is a problem that can be solved in polynomial time \citep{CsKa17a}.\footnote{In fact, the algorithm by \citet{CsKa17a} for solving \PopEdge uses \DomEdge as an important subroutine.}

	\begin{wbox}
		\DomEdge\\
		\textbf{Input:} Instance $\ins$ of \MatP and designated edge $e\in E^{\ins}$.\\
		\textbf{Task:} Compute a dominant matching in $\ins$ containing $e$ or determine that there exists no such matching.
	\end{wbox}

Finally, we introduce two \NP-complete problems that we employ in our hardness results.
They are both related to finding stable matchings containing or not containing certain edges and vertices.
The first one is the problem of finding a popular matching that avoids two designated edges.
The problem is known to be \NP-complete \citep[Theorem~4.1]{FKPZ19a}.\footnote{The validity of the restriction that the two designated edges can be assumed to be disjoint immediately follows from the proof by \citet{FKPZ19a}.}

	\begin{wbox}
		\ForbEdge\\
		\textbf{Input:} Instance $\ins$ of \MatP and two designated edges $e,e'\in E^{\ins}$ where $e\cap e' = \emptyset$.\\
		\textbf{Task:} Compute a popular matching in $\ins$ not containing $e$ and $e'$, or determine that there exists no such matching.
	\end{wbox}

Moreover, we consider the problem of finding a popular matching that avoids one designated edge while matching a designated vertex.
The following restricted variant of this problem is known to be \NP-complete \citep[Theorem~1]{FPZ18a}.\footnote{The same result is also contained in the conference version of their paper \citep{FKPZ19a}, however, the construction given their does not satisfy the additional restrictions that we want to apply.}

	\begin{wbox}
		\ForbEdgeForceVert\\
		\textbf{Input:} Instance $\ins$ of \MatP, a designated edge $e = \{a,b\}\in E^{\ins}$, and a vertex $d\in W\cup F$, where $N_a^{\ins} = \{b\}$, $b$ only contains one neighbor~$c$ other than~$a$, and $b$ and~$c$ top-rank each other.\\
		\textbf{Task:} Compute a popular matching in $\ins$ not containing $e$ and covering $d$, or determine that there exists no such matching.
	\end{wbox}

Finally, we will use the following \NP-hard version of \tSAT.
    \begin{wbox}
		\mSAT\\
		\textbf{Input:} A \tSAT\ formula $\varphi = C_1\wedge \dots \wedge C_m$ over variables $\{ X_1,\dots, X_n\}$, satisfying that for all clauses $C_j=(\ell_{j_1}\vee \ell_{j_2}\vee \ell_{j_3})$, either all literals $\ell_{j_i}$ are positive, or all are negative. \\
		\textbf{Task:} Find a satisfying truth assignment $\Phi:[n]\to \{ \mathrm{True},\mathrm{False}\}$ or conclude that none exists.
	\end{wbox}
\NP-hardness of \mSAT\ was shown by~\cite{FKPZ19a}. While their hardness reduction contained clauses with less than 3 literals, we can assume that all clauses contain exactly 3 literals by duplicating literals within clauses if needed (e.g., we take $(X_1\vee X_2\vee X_2)$ instead of $(X_1\vee X_2)$).

\section{Results}

In this section, we present our results.

\subsection{Perturbations of One Agent}\label{sec:oneagent}

First, we consider instances of \RobustProb and \RobustDom with identical underlying graphs, where the perturbed instance only differs with respect to the preferences of a single agent.
We will eventually show that both problems admit a polynomial-time algorithm. 
For this, we perform two key steps.
First, we define a set of hybrid instances, which allow us to answer if there exists a robust popular matching that contains a given edge.
Second, we deal with the case of robust popular matchings where the agent with perturbed preferences remains unmatched.
Combining these insights with known algorithmic and structural results about popular matchings, we obtain a polynomial-time algorithm.

We start by defining our hybrid instances.
Consider a pair of \MatP instances $\iP$ 
where $\iB$ only differs from $\iA$ with respect to the preferences of agent~$x$.
Let $G = (W\cup F, E)$ be the underlying graph and consider an edge $e\in E$ with $x\in E$, say $e = \{x,y\}$.
Define $P^A = \{z\in W\cup F\colon z\succ_x^{\iA} y\}$ and $P^B = \{z\in W\cup F\colon z\succ_x^{\iB} y\}$, i.e., $P^A$ and $P^B$ are the agents preferred to $y$ by $x$ in instances $\iA$ and $\iB$, respectively.
Consider any linear order $\succ'$ of the neighbors $N_x$ of $x$ in $G$ that satisfies $z \succ' y$ if $z\in P^A\cup P^B$, as well as $y\succ' z$ if $z\in N_x\setminus (P^A\cup P^B\cup \{y\})$.
Hence, $\succ'$ is a preference order, where $P^A\cup P^B$ is ordered arbitrarily at the top, then agent~$y$, and finally all other neighbors of $x$ in an arbitrary order.

The \emph{hybrid instance} $\iH$ of $\iP$ with respect to $e$ is defined as the instance of {\MatP} where $\succ^{\iH}_z = \succ^{\iA}_z$ for all $z\neq x$ and $\succ^{\iH}_x = \succ'$.
Note that we illustrate hybrid instances later on in \Cref{ex:proofhybrid}, where we also illustrate our main proof.
We now prove two important lemmas that create a correspondence of popular matchings in $\iH$ and robust popular matchings for $\iP$.
The first lemma considers popular matchings in $\iH$ containing~$e$.

\begin{lemma}
    \label{lem:hybrid_complete_1}
    Let $M$ be a matching and $e\in M$ with $x\in e$.
    \begin{enumerate}
        \item If $M$ is popular for $\iH$, then it is robust popular with respect to $\iA$ and $\iB$.
        \item If $M$ is dominant for $\iH$, then it is robust dominant with respect to $\iA$ and $\iB$.
    \end{enumerate}
\end{lemma}
\begin{proof}
    Let $M$ be a matching and $e\in M$ with $x\in e$.
    Let $M'$ be any other matching and let $X \in \{A,B\}$ specify an instance.
    We want to relate the popularity margin between $M$ and $M'$ in $\iX$ with their popularity margin in $\iH$.
    First, let $z\in (W\cup F)\setminus \{x\}$.
    Then, since the preferences of $z$ are identical in $\iH$ and $\iX$, it holds that
    $\vote_z^{\iX}(M,M') = \vote_z^{\iH}(M,M')$.
    Second, let us consider the vote of agent~$x$.
    By construction of the hybrid instance, for all agents $z\in N_x$, it holds that
    $y\succ^{\iX}_x z$ whenever $y\succ^{\iH}_x z$.
    Hence, since $M(x) = y$, we can conclude that
    $\vote_x^{\iX}(M,M') \ge \vote_x^{\iH}(M,M')$.
    
    Combining these two insights, we obtain
    \begin{align*}
        &\pmarg^{\iX}(M,M') = \sum_{z\in W\cup F} \vote_z^{\iX}(M,M')\\ &\ge \sum_{z\in W\cup F} \vote_z^{\iH}(M,M') = \pmarg^{\iH}(M,M')\text.
    \end{align*}
    Hence, $\pmarg^{\iH}(M,M') \ge 0$ implies $\pmarg^{\iA}(M,M') \ge 0$ and $\pmarg^{\iB}(M,M') \ge 0$.
    Therefore, whenever $M$ is popular for $\iH$, then it is popular for both $\iA$ and $\iB$.
    
    Similarly, $\pmarg^{\iH}(M,M') > 0$ implies $\pmarg^{\iA}(M,M') > 0$ and $\pmarg^{\iB}(M,M') > 0$ and hence $M$ is dominant for both $\iA$ and $\iB$ whenever it is dominant for $\iH$.
\end{proof}

The next lemma shows that the converse is also true.

\begin{lemma}
    \label{lem:hybrid_complete_2}
    Let $M$ be a matching and $e\in M$ with $x\in e$.
    \begin{enumerate}
        \item If $M$ is robust popular with respect to~$\iA$ and~$\iB$, then it is popular for~$\iH$.
        \item If $M$ is robust dominant with respect to~$\iA$ and~$\iB$, then it is dominant for~$\iH$. 
    \end{enumerate}
    
\end{lemma}
\begin{proof}
    Let $M$ be a matching and $e\in M$ with $x\in e$.
    Let $M'$ be any other matching.
    We will compute the popularity margin between $M$ and $M'$ in~$\iH$.
    
    Let $z\in (W\cup F)\setminus \{x\}$.
    As in the proof of the previous lemma, since the preferences of $z$ are identical in $\iH$, $\iA$, and $\iB$, it holds that
    $\vote_z^{\iH}(M,M') = \vote_z^{\iA}(M,M') = \vote_z^{\iB}(M,M')$.
    
    We make a case distinction with respect to the vote of agent~$x$.
    If $\vote_x^{\iH}(M,M') = 1$, then the previous observation immediately implies that $\pmarg^{\iH}(M,M') \ge \pmarg^{\iA}(M,M') \ge 0$.
    If $\vote_x^{\iH}(M,M') = 0$, i.e., $M'(x) = M(x)$, then $\pmarg^{\iH}(M,M') = \pmarg^{\iA}(M,M')$.
    If $\vote_x^{\iH}(M,M') = -1$, then $M'(x) \succ^{\iH}_x M(x)$ where $M(x) = y$, and, therefore, $M'(x) \in P^A\cup P^B$.
    Let $X\in \{A,B\}$ with $M'(x)\in P^X$.
    Then, by definition, $M'(x) \succ^{\iX}_x M(x)$, and, therefore, $\vote_x^{\iX}(M,M') = -1$.
    Combining this with the votes of the other agents, it follows that $\pmarg^{\iH}(M,M') = \pmarg^{\iX}(M,M')$.

    Hence, we have shown that, for every $M'$, there exists $X\in \{A,B\}$ with $\pmarg^{\iH}(M,M') = \pmarg^{\iX}(M,M')$.
    Hence, $M$ is popular for $\iH$ if $M$ is robust popular with respect to $\iA$ and $\iB$, and $M$ is dominant for $\iH$ if $M$ is robust dominant with respect to $\iA$ and $\iB$.
\end{proof}

Combining \Cref{lem:hybrid_complete_1,lem:hybrid_complete_2}, we can find robust popular and robust dominant matchings containing a specific edge by solving an instance of \PopEdge or \DomEdge, respectively.

\begin{corollary}\label{cor:robustedge}
    The following are true:
    \begin{enumerate}
        \item The instance $\iP$ contains a robust popular matching containing edge $e$ if and only if \PopEdge for the hybrid instance with designated edge $e$ is a Yes-instance.
        \item The instance $\iP$ contains a robust dominant matching containing edge $e$ if and only if \DomEdge for the hybrid instance with designated edge $e$ is a Yes-instance.
    \end{enumerate}
\end{corollary}

It remains to figure out whether there exists a robust popular and robust dominant matching that leaves the agent with perturbed preferences unmatched.
For this, we make another observation.

\begin{lemma}\label{lem:unmatched}
    Let $M$ be a matching that leaves agent $x$ unmatched. 
    Then, 
    \begin{enumerate}
        \item $M$ is popular for $\iA$ if and only if $M$ is popular for $\iB$, and 
        \item $M$ is dominant for $\iA$ if and only if $M$ is dominant for $\iB$.
    \end{enumerate}
\end{lemma}
\begin{proof}
    Let $M$ be a matching that leaves~$x$ unmatched.
    Then, for every matching $M'$, it holds that $\vote^{\iA}_x(M,M') = \vote^{\iB}_x(M,M')$.
    Hence, since $x$ is the only agent to perturb their preferences, it follows that $\pmarg^{\iA}(M,M') = \pmarg^{\iB}(M,M')$.
    Therefore, as $M'$ was an arbitrary matching, it holds that $M$ is popular (or dominant) for $\iA$ if and only if $M$ is popular (or dominant) for $\iB$.
\end{proof}

As a consequence, we can tackle this case by finding a popular matching in $\iA$ that leaves $x$ unmatched, or decide that no such matching exists.
This problem has a surprisingly easy solution:
It suffices to compute any stable matching.
The key insight is captured in the next lemma by \citet{CsKa17a}, a lemma that resembles the fundamental Rural Hospitals Theorem for stable matchings \citep{GaSo85b,Roth84b}.

\begin{lemma}[\citet{CsKa17a}]\label{lem:rural_popular}
    If an agent is unmatched in some popular matching, then it is unmatched in all stable matchings.
\end{lemma}

Moreover, for dominant matchings, we can check a dominant matching instead of a stable matching.
This follows directly from the results by \citet{CsKa17a}, but since they do not explicitly state this, we include a short proof.

\begin{lemma}\label{lem:dominant_rural_popular}
    All dominant matchings cover the same set of agents.
    Hence, if an agent is unmatched in some dominant matching, then it is unmatched in all dominant matchings.
\end{lemma}

\begin{proof}
    \citet{CsKa17a} provide a reduction that shows a correspondence of dominant matchings in the original instance and stable matchings in a reduced instance.
    In their mapping, a worker $w$ in the original instance corresponds to two copies $w_0$ and $w_1$ in the reduced instance.
    In particular, $w$ is covered in a dominant matching in the original instance if and only if both $w_0$ and $w_1$ are covered in a stable matching in the reduced instance. 
    Moreover, a firm $f$ corresponds to a single copy $f_0$ and $f$ is covered by a dominant matching if and only if $f'$ is covered by a stable matching. 
    Since all stable matchings cover the same set of agents, it follows that dominant matchings in the original instance cover the same set of agents.
\end{proof}

We combine all our insights to state an algorithm for \RobustProb and \RobustDom if the perturbed input instance only differs with respect to the preference order of one agent.
We state the algorithm for robust popularity and highlight the changes for robust dominance in parenthesis.

\begin{algorithm}
\caption{Solving \RobustProb and \RobustDom under preference changes of one agent}\label{alg:robust}
  \begin{flushleft}
    \textbf{Input:} Pair $\iP$ of \MatP instances\\ 
    \textbf{Output:} Robust popular (or dominant) matching for $\iP$ or statement that no such matching exists
  \end{flushleft}
\begin{algorithmic}[1]
\State Compute stable (or dominant) matching $M$ for $\iA$
\If{$M$ leaves $x$ unmatched}
    \Return $M$.
\EndIf
\For{$e\in E$ with $x\in e$}
    \If{there exists a popular (or dominant) matching $M$ for $\iH$ with $e\in M$}
        \Return $M$
    \EndIf
\EndFor
\State\Return ``No robust popular (or dominant) matching exists''
\end{algorithmic}
\end{algorithm}

The algorithm first checks a stable (or dominant) matching to attempt finding a robust popular (or dominant) matching that leaves $x$ unmatched.
Then, it checks the hybrid instances to search for robust popular (or dominant) matchings where $x$ is matched.
The correctness and running time of this algorithm are captured in the main theorem of this section.

\begin{theorem}
    \RobustProb and \RobustDom can be solved in polynomial time if the perturbed input instance only differs with respect to the preference order of one agent.
\end{theorem}

\begin{proof}
    The polynomial running time follows because \StabMat, \PopEdge, and \DomEdge can be solved in polynomial time \citep{GaSh62a,CsKa17a}.

    Let us consider the correctness of \Cref{alg:robust}.
    First, note that, in the case of popularity, if \Cref{alg:robust} returns a matching in line~2, then it returns a popular matching for $\iA$ because stable matchings are popular \citep{Gard75a}.
    Hence, by \Cref{lem:unmatched}, it returns a robust popular (or dominant) matching in this case.
    Moreover, if \Cref{alg:robust} returns a matching in line~5, it is a robust popular (or dominant) matching according to \Cref{cor:robustedge}.
    Hence, if \Cref{alg:robust} returns a matching, then it is a robust popular (or dominant) matching.
    Thus, it remains to show that \Cref{alg:robust} returns a matching if there exists a robust popular (or dominant) matching in the considered instance.
    
    Assume, therefore, that there exists a robust popular (or dominant) matching $M$.
    Assume first that $M$ leaves $x$ unmatched.
    Then, in the case of popularity, by \Cref{lem:rural_popular}, every stable matching leaves $x$ unmatched, and \Cref{alg:robust} returns a matching in line~2.
    In the case of dominance, \Cref{lem:dominant_rural_popular} implies that all dominant matchings match the same vertices and hence a matching leaving $x$ unmatched is returned in line~2 as well.
    In addition, if $x$ is matched in $M$ by an edge $e$, then, by \Cref{cor:robustedge}, $\iH$ contains a popular (or dominant) matching containing~$e$.
    Hence, \Cref{alg:robust} returns a matching in line~5.
\end{proof}

We illustrate the proof as well as hybrid instances by continuing \Cref{ex:basic}.

\begin{example}\label{ex:proofhybrid}
    
    Recall that, in \Cref{ex:basic}, the two input instances $\iA$ and $\iB$ only differ with respect to the preferences of agent~$w_1$.
    Moreover, we have already seen that $M_1 = \{\{w_1,f_1\}, \{w_2,f_3\}, \{w_3,f_2\}, \{w_4,f_4\}\}$ is stable (and, therefore, popular) for $\iA$ but not popular for $\iB$ and that $M_2 = \{\{w_1,f_3\}, \{w_2,f_1\}, \{w_3,f_2\}, \{w_4,f_4\}\}$ is robust popular.
    
    Since the stable matching $M_1$ for $\iA$ matches $w_1$, we can conclude that there exists no robust popular matching that leaves $w_1$ unmatched.
    Hence, we have to consider the hybrid instances $\iH$ for $e\in \{\{w_1,f_1\},\{w_1,f_2\},\{w_1,f_3\}\}$, i.e., all edges incident to $w_1$.
    Interestingly, $\iA$ can serve as a hybrid instance for $e \in \{\{w_1,f_2\},\{w_1,f_3\}\}$ and $\iB$ can serve as a hybrid instance for $e = \{w_1,f_1\}$.
    In fact, this follows from a more general observation concerning hybrid instances for an instance pair $(\iA,\iB)$ where $\iB$ evolves from $\iA$ by a downshift of agent~$y$ in the preference order of agent~$x$.
    Whenever this is the case, $\iB$ serves as a hybrid instance for $e = \{x,y\}$, while $\iA$ serves as the hybrid instance for all other edges containing~$x$.
    
    Now, since $M_2$ is popular for $\iA$, it is a popular matching containing $e$ for the hybrid instance~$\iH$ with $e = \{w_1,f_3\}$.
    Hence, \Cref{alg:robust} finds the robust popular matching for $\iP$.\hfill$\lhd$
\end{example}

Finally, by straightforward extensions of the techniques developed in this section, we can generalize our result for the case of more than two instances that all differ only with respect to the preferences of one agent~$x$.
To find a robust popular (or dominant) matching containing a specific edge $e = \{x,y\}$, we define the preference order of~$x$ in a generalized hybrid instance 
by putting the agents preferred to $y$ by $x$ in \emph{any} input instance above $y$. 
This ensures that whenever we contest the popularity (or dominance) of a matching in the hybrid instance with a matching where~$x$ receives a better partner~$z$, then the popularity (or dominance) of this matching is also contested in the input instances that have $z$ ranked above $y$.

\begin{theorem}\label{thm:multiple_instances}
    There exists a polynomial-time algorithm for the following problem:
    Given a collection of {\MatP} instances $(\ins_1,\dots, \ins_k)$, which are all defined for the same underlying graph and differ only with respect to the preferences of a single agent, does there exist a matching that is popular (or dominant) for $\ins_i$ for all $1\le i\le k$?
\end{theorem}

\subsection{Perturbations of Two Agents}\label{sec:2hard}

In this section, we continue the consideration of instance pairs with the identical underlying graph.
While we have previously seen a polynomial-time algorithm for solving \RobustProb and \RobustDom if the perturbed instance only differs by a single agent permuting their preferences, we now allow several agents to change their preference orders.
In this case, we obtain two striking hardness results:
The first one holds for \RobustProb and \RobustDom and the case when two agents from the same side permuting their preferences. 
Inspecting the proof, one can see that one of these agents only has two neighbors and therefore has only two ways to represent their preferences.
The other agent has several neighbors, but their preference perturbation can be designed to be a preference reversal.
The second hardness result investigates when two agents only perform the simplest perturbation operation of a swap. 
We obtain a hardness result for \RobustProb, when one worker and one firm each perform a swap.
Both results demonstrate a complexity dichotomy with respect to the number of agents perturbing their preferences.

\subsubsection{Perturbations of Two Same-Type Agents}

The proof of our first hardness result is inspired by the hardness result of \ForbEdge\ by~\cite{FKPZ19a}.

\begin{theorem}\label{thm:twooneside}
    \RobustProb\ and \RobustDom\ are \NP-complete, even if $\iA$ and $\iB$ differ for only two agents on the same side. 
\end{theorem}
\begin{proof}
Membership in \NP{} is straightforward, as we can verify if a matching $M$ is popular or dominant in polynomial time. 
Indeed, a robust popular (or robust dominant) matching with respect to two given input instances of {\MatP} serves as a polynomial-size certificate for a Yes-instance.
We can verify it by simply checking whether the matching is popular (or dominant) in both instances in polynomial time (\citealp[Theorem~9]{BIM10a}; \citealp[Lemma~1]{CsKa17a}).

    For \NP-hardness, we perform a reduction from \mSAT, as defined in \Cref{sec:problems}.
    For this, consider an instance $\ins$ of \mSAT\ with $m$ clauses given by $\varphi = C_1\wedge \dots \wedge C_m$ and $n$ variables $\{ X_1,\dots, X_n\}$.

    We create an instance pair $\iP$ of \RobustProb\ and \RobustDom\ (we use the same reduction for both). The reduction is illustrated in \Cref{fig:constrCombined} and formally described as follows.
    
\begin{figure}
    \centering
    
	\newcommand{\rankcolor}{black}
	\newcommand{\iAcolor}{myblue}
	\newcommand{\iBcolor}{myred}

	\resizebox{1\textwidth}{!}{
	\begin{tikzpicture}
	  \pgfmathsetmacro{\vertexsize}{0.6}
	  \pgfmathsetmacro{\gadgetrad}{1.2}
	  \pgfmathsetmacro{\edgedist}{.8}
	  \pgfmathsetmacro{\gadgetdist}{3.7}
	  \pgfmathsetmacro{\gadgetstretch}{7}

	  \foreach \k/\mult in {1/2} {
	    \pgfmathsetmacro{\y}{\mult * \gadgetdist}
	    \foreach \l/\x in {1/0,2/\gadgetstretch, 3/2*\gadgetstretch} {
	      \coordinate (p\l\k) at  (\x,\y) {};
	      \node[draw, rectangle, inner sep=0, minimum size=\vertexsize cm, minimum height=\vertexsize cm] (b\l\k) at ($(p\l\k)+(-\gadgetrad,0)$) {$b_{\l}^{\k}$};
	      \node[draw, rectangle, inner sep=0, minimum size=\vertexsize cm, minimum height=\vertexsize cm] (e\l\k) at ($(p\l\k)+(0,\gadgetrad)$) {$e_{\l}^{\k}$};
	      \node[draw, circle, inner sep=0, minimum size=\vertexsize cm, minimum height=\vertexsize cm] (c\l\k) at ($(p\l\k)+(0,-\gadgetrad)$) {$c_{\l}^{\k}$};
	      \node[draw, circle, inner sep=0, minimum size=\vertexsize cm, minimum height=\vertexsize cm] (d\l\k) at ($(p\l\k)+(\gadgetrad,0)$) {$d_{\l}^{\k}$};
	      \node[draw, rectangle, inner sep=0, minimum size=\vertexsize cm, minimum height=\vertexsize cm] (f\l\k) at ($(p\l\k)+(\gadgetstretch/2,0)+(-\edgedist,0)$) {$f_{\l}^{\k}$};
	      \node[draw, circle, inner sep=0, minimum size=\vertexsize cm, minimum height=\vertexsize cm] (a\l\k) at ($(p\l\k)+(\gadgetstretch/2,0)+(\edgedist,0)$) {$a_{\l}^{\k}$};

	      \draw (b\l\k) edge node[yshift = -.15cm,pos = .1] {\color{\rankcolor}\footnotesize $3$} node[yshift = -.15cm,pos = .95] {\color{\rankcolor}\footnotesize $1$} (d\l\k);
	      \draw (c\l\k) edge node[xshift =.15cm,pos = .1] {\color{\rankcolor}\footnotesize $1$} node[xshift =.15cm,pos = .9] {\color{\rankcolor}\footnotesize $2$} (e\l\k);
	      \draw (d\l\k) edge node[xshift =.1cm,yshift = .1cm, pos = .1] {\color{\rankcolor}\footnotesize $2$} node[xshift =.1cm,yshift = .1cm,pos = .9] {\color{\rankcolor}\footnotesize $1$} (e\l\k);
	      \draw (d\l\k) edge node[yshift = -.15cm,pos = .1] {\color{\rankcolor}\footnotesize $3$} node[yshift = -.15cm,pos = .9] {\color{\rankcolor}\footnotesize $1$} (f\l\k);
	      \draw (f\l\k) edge node[yshift = -.15cm,pos = .1] {\color{\rankcolor}\footnotesize $2$} node[yshift = -.15cm,pos = .92] {\color{\rankcolor}\footnotesize $1$} (a\l\k);
      
	    }
    
	      \draw (b1\k) edge node[below,pos = -.1] {\color{\rankcolor}\footnotesize $1$} node[below,pos = .7] {\color{\rankcolor}\footnotesize $3$} (c1\k);
	      \draw (b2\k) edge node[below,pos = -.1] {\color{\rankcolor}\footnotesize $1$} node[below,pos = .7] {\color{\rankcolor}\footnotesize $3$} (c2\k);
	      \draw (b3\k) edge node[below,pos = -.1] {\color{\rankcolor}\footnotesize $1$} node[below,pos = .7] {\color{\rankcolor}\footnotesize $2$} (c3\k);
	      \draw (a1\k) edge node[yshift = -.15cm,pos = .1] {\color{\rankcolor}\footnotesize $2$} node[yshift = -.15cm,pos = .85] {\color{\rankcolor}\footnotesize $2$} (b2\k);
	      \draw (a2\k) edge node[yshift = -.15cm,pos = .1] {\color{\rankcolor}\footnotesize $2$} node[yshift = -.15cm,pos = .85] {\color{\rankcolor}\footnotesize $2$} (b3\k);
    
	  }

	  \foreach \k/\mult in {2/1} {
	    \pgfmathsetmacro{\y}{\mult * \gadgetdist}
	    \foreach \l/\x in {1/0,2/\gadgetstretch, 3/2*\gadgetstretch} {
	      \coordinate (p\l\k) at  (\x,\y) {};
	      \node[draw, rectangle, inner sep=0, minimum size=\vertexsize cm, minimum height=\vertexsize cm] (b\l\k) at ($(p\l\k)+(225:\gadgetrad)$) {$\bar b_{\l}^{\k}$};
	      \node[draw, rectangle, inner sep=0, minimum size=\vertexsize cm, minimum height=\vertexsize cm] (e\l\k) at ($(p\l\k)+(135:\gadgetrad)$) {$\bar e_{\l}^{\k}$};
	      \node[draw, circle, inner sep=0, minimum size=\vertexsize cm, minimum height=\vertexsize cm] (c\l\k) at ($(p\l\k)+(45:\gadgetrad)$) {$\bar c_{\l}^{\k}$};
	      \node[draw, circle, inner sep=0, minimum size=\vertexsize cm, minimum height=\vertexsize cm] (d\l\k) at ($(p\l\k)+(315:\gadgetrad)$) {$\bar d_{\l}^{\k}$};
	      \node[draw, rectangle, inner sep=0, minimum size=\vertexsize cm, minimum height=\vertexsize cm] (f\l\k) at ($(p\l\k)+(\gadgetstretch/2,0)+(-\edgedist,0)$) {$\bar f_{\l}^{\k}$};
	      \node[draw, circle, inner sep=0, minimum size=\vertexsize cm, minimum height=\vertexsize cm] (a\l\k) at ($(p\l\k)+(\gadgetstretch/2,0)+(\edgedist,0)$) {$\bar a_{\l}^{\k}$};

	      \draw (b\l\k) edge node[xshift =-.1cm,yshift = .1cm,pos = .1] {\color{\rankcolor}\footnotesize $2$} node[xshift =-.1cm,yshift = .1cm,pos = .9] {\color{\rankcolor}\footnotesize $2$} (c\l\k);
	      \draw (c\l\k) edge node[yshift =.15cm,pos = .1] {\color{\rankcolor}\footnotesize $1$} node[yshift =.15cm,pos = .9] {\color{\rankcolor}\footnotesize $3$} (e\l\k);
	      \draw (d\l\k) edge node[xshift =.1cm,yshift = .1cm, pos = .1] {\color{\rankcolor}\footnotesize $4$} node[xshift =.1cm,yshift = .1cm,pos = .94] {\color{\rankcolor}\footnotesize $2$} (e\l\k);
	      \draw (d\l\k) edge node[yshift = -.15cm,pos = .1] {\color{\rankcolor}\footnotesize $3$} node[yshift = -.15cm,pos = .9] {\color{\rankcolor}\footnotesize $2$} (f\l\k);
	      \draw (f\l\k) edge node[yshift = -.15cm,pos = .1] {\color{\rankcolor}\footnotesize $1$} node[yshift = -.15cm,pos = .92] {\color{\rankcolor}\footnotesize $1$} (a\l\k);
      
	    }
	      \draw (b1\k) edge node[yshift = .15cm,pos = .1] {\color{\rankcolor}\footnotesize $4$} node[yshift = .15cm,pos = .95] {\color{\rankcolor}\footnotesize $2$} (d1\k);
	      \draw (b3\k) edge node[yshift = .15cm,pos = .1] {\color{\rankcolor}\footnotesize $4$} node[yshift = .15cm,pos = .95] {\color{\rankcolor}\footnotesize $2$} (d3\k);
	      \draw (b2\k) edge node[yshift = .15cm,pos = .1] {\color{\rankcolor}\footnotesize $5$} node[yshift = .15cm,pos = .95] {\color{\rankcolor}\footnotesize $2$} (d2\k);
	      \draw (a1\k) edge node[yshift = -.15cm,pos = .1] {\color{\rankcolor}\footnotesize $2$} node[yshift = -.15cm,pos = .92] {\color{\rankcolor}\footnotesize $1$} (b2\k);
	      \draw (a2\k) edge node[yshift = -.15cm,pos = .1] {\color{\rankcolor}\footnotesize $2$} node[yshift = -.15cm,pos = .92] {\color{\rankcolor}\footnotesize $1$} (b3\k);
	      \draw (a1\k) edge node[yshift = .15cm,pos = .05] {\color{\rankcolor}\footnotesize $3$} node[yshift = .15cm,pos = .92] {\color{\rankcolor}\footnotesize $1$} (e2\k);
	      \draw (a2\k) edge node[yshift = .15cm,pos = .05] {\color{\rankcolor}\footnotesize $3$} node[yshift = .15cm,pos = .92] {\color{\rankcolor}\footnotesize $1$} (e3\k);
	  }

	  \foreach \k/\mult in {3/0} {
	    \pgfmathsetmacro{\y}{\mult * \gadgetdist}
	    \foreach \l/\x in {1/0,2/\gadgetstretch, 3/2*\gadgetstretch} {
	      \coordinate (p\l\k) at  (\x,\y) {};
	      \node[draw, rectangle, inner sep=0, minimum size=\vertexsize cm, minimum height=\vertexsize cm] (b\l\k) at ($(p\l\k)+(-\gadgetrad,0)$) {$b_{\l}^{\k}$};
	      \node[draw, rectangle, inner sep=0, minimum size=\vertexsize cm, minimum height=\vertexsize cm] (e\l\k) at ($(p\l\k)+(0,-\gadgetrad)$) {$e_{\l}^{\k}$};
	      \node[draw, circle, inner sep=0, minimum size=\vertexsize cm, minimum height=\vertexsize cm] (c\l\k) at ($(p\l\k)+(0,\gadgetrad)$) {$c_{\l}^{\k}$};
	      \node[draw, circle, inner sep=0, minimum size=\vertexsize cm, minimum height=\vertexsize cm] (d\l\k) at ($(p\l\k)+(\gadgetrad,0)$) {$d_{\l}^{\k}$};
	      \node[draw, rectangle, inner sep=0, minimum size=\vertexsize cm, minimum height=\vertexsize cm] (f\l\k) at ($(p\l\k)+(\gadgetstretch/2,0)+(-\edgedist,0)$) {$f_{\l}^{\k}$};
	      \node[draw, circle, inner sep=0, minimum size=\vertexsize cm, minimum height=\vertexsize cm] (a\l\k) at ($(p\l\k)+(\gadgetstretch/2,0)+(\edgedist,0)$) {$a_{\l}^{\k}$};

	      \draw (b\l\k) edge node[yshift = .15cm,pos = .1] {\color{\rankcolor}\footnotesize $3$} node[yshift = .15cm,pos = .95] {\color{\rankcolor}\footnotesize $1$} (d\l\k);
	      \draw (c\l\k) edge node[xshift =.15cm,pos = .1] {\color{\rankcolor}\footnotesize $1$} node[xshift =.15cm,pos = .9] {\color{\rankcolor}\footnotesize $2$} (e\l\k);
	      \draw (d\l\k) edge node[yshift = -.2cm, pos = 0] {\color{\rankcolor}\footnotesize $2$} node[yshift = -.2cm,pos = .85] {\color{\rankcolor}\footnotesize $1$} (e\l\k);
	      \draw (d\l\k) edge node[yshift = -.15cm,pos = .1] {\color{\rankcolor}\footnotesize $3$} node[yshift = -.15cm,pos = .9] {\color{\rankcolor}\footnotesize $1$} (f\l\k);
	      \draw (f\l\k) edge node[yshift = -.15cm,pos = .1] {\color{\rankcolor}\footnotesize $2$} node[yshift = -.15cm,pos = .92] {\color{\rankcolor}\footnotesize $1$} (a\l\k);
      
	    }
    
	      \draw (b1\k) edge node[yshift = .2cm,pos = 0] {\color{\rankcolor}\footnotesize $1$} node[yshift = .2cm,pos = .75] {\color{\rankcolor}\footnotesize $2$} (c1\k);
	      \draw (b2\k) edge node[yshift = .2cm,pos = 0] {\color{\rankcolor}\footnotesize $1$} node[yshift = .2cm,pos = .75] {\color{\rankcolor}\footnotesize $3$} (c2\k);
	      \draw (b3\k) edge node[yshift = .2cm,pos = 0] {\color{\rankcolor}\footnotesize $1$} node[yshift = .2cm,pos = .75] {\color{\rankcolor}\footnotesize $3$} (c3\k);
	      \draw (a1\k) edge node[yshift = -.15cm,pos = .1] {\color{\rankcolor}\footnotesize $2$} node[yshift = -.15cm,pos = .85] {\color{\rankcolor}\footnotesize $2$} (b2\k);
	      \draw (a2\k) edge node[yshift = -.15cm,pos = .1] {\color{\rankcolor}\footnotesize $2$} node[yshift = -.15cm,pos = .85] {\color{\rankcolor}\footnotesize $2$} (b3\k);
    
	  }

	  \draw (c11) edge[bend left = 17] node[xshift = .1cm, yshift = .2cm,pos = 0] {\color{\rankcolor}\footnotesize $2$} node[yshift = .2cm,pos = 1] {\color{\rankcolor}\footnotesize $3$} (b22);
	  \draw (c21) edge[bend left = 17] node[xshift = .1cm, yshift = .2cm,pos = 0] {\color{\rankcolor}\footnotesize $2$} node[yshift = .2cm,pos = 1] {\color{\rankcolor}\footnotesize $3$} (b32);
  
	  \draw (b12) edge node[xshift = .1cm, yshift = -.2cm,pos = 0] {\color{\rankcolor}\footnotesize $3$} node[xshift = -.1cm, yshift = .2cm,pos = 1] {\color{\rankcolor}\footnotesize $2$} (c23);
	  \draw (b22) edge node[xshift = .1cm, yshift = -.2cm,pos = 0] {\color{\rankcolor}\footnotesize $4$} node[xshift = -.1cm, yshift = .2cm,pos = 1] {\color{\rankcolor}\footnotesize $2$} (c33);

	    \node[draw, rectangle, inner sep=0, minimum size=\vertexsize cm, minimum height=\vertexsize cm] (v1) at ($(a32)+(2*\edgedist,0)$) {$v_1$};
	    \node[draw, circle, inner sep=0, minimum size=\vertexsize cm, minimum height=\vertexsize cm] (w1) at ($(v1)+(0,-2*\edgedist)$) {$w_1$};
	    \node[draw, rectangle, inner sep=0, minimum size=\vertexsize cm, minimum height=\vertexsize cm] (v2) at ($(v1)+(0,-4*\edgedist)$) {$v_2$};
	    \node[draw, circle, inner sep=0, minimum size=\vertexsize cm, minimum height=\vertexsize cm] (w2) at ($(v1)+(0,-6*\edgedist)$) {$w_2$};  

	    \node[draw, circle, inner sep=0, minimum size=\vertexsize cm, minimum height=\vertexsize cm] (t1) at ($(b12)!.5!(e12)+(-2*\edgedist,0)$) {$t_1$};
	    \node[draw, rectangle, inner sep=0, minimum size=\vertexsize cm, minimum height=\vertexsize cm] (s1) at ($(t1)+(0,-2*\edgedist)$) {$s_1$};
	    \node[draw, circle, inner sep=0, minimum size=\vertexsize cm, minimum height=\vertexsize cm] (t2) at ($(t1)+(0,-4*\edgedist)$) {$t_2$};
	    \node[draw, rectangle, inner sep=0, minimum size=\vertexsize cm, minimum height=\vertexsize cm] (s2) at ($(t1)+(0,-6*\edgedist)$) {$s_2$};
	    \node[draw, circle, inner sep=0, minimum size=\vertexsize cm, minimum height=\vertexsize cm] (t3) at ($(t1)+(0,-8*\edgedist)$) {$t_3$};
     
	      \draw (a31) edge node[xshift = .2cm, yshift = -.1cm,pos = 0] {\color{\rankcolor}\footnotesize $2$} node[xshift = .05cm, yshift = .2cm,pos = 1] {\color{\rankcolor}\footnotesize $1$} (v1);
	      \draw (a32) edge node[xshift = .1cm, yshift = .15cm,pos = 0] {\color{\rankcolor}\footnotesize $2$} node[xshift = -.1cm, yshift = .15cm,pos = 1] {\color{\rankcolor}\footnotesize $2$} (v1);
	      \draw (a33) edge node[xshift = -.1cm, yshift = .2cm,pos = 0] {\color{\rankcolor}\footnotesize $2$} node[xshift = -.2cm, yshift = -.15cm,pos = 1] {\color{\rankcolor}\footnotesize $3$} (v1);

	       \draw (v1) edge node[xshift = .15cm, yshift = -.15cm,pos = 0] {\color{\rankcolor}\footnotesize $4$} node[xshift = .15cm, yshift = .15cm,pos = 1] {\color{\rankcolor}\footnotesize $2$} (w1);
	       \draw (w1) edge node[xshift = .15cm, yshift = -.15cm,pos = 0] {\color{\rankcolor}\footnotesize $1$} node[xshift = -.15cm, yshift = .15cm,pos = 1] {\color{\iAcolor}\footnotesize $1$} node[xshift = .15cm, yshift = .15cm,pos = 1] {\color{\iBcolor}\footnotesize $\ell$} (v2);
	       \draw (v2) edge node[xshift = -.15cm, yshift = -.15cm,pos = 0] {\color{\iAcolor}\footnotesize $2$} node[xshift = .4cm, yshift = -.15cm,pos = 0] {\color{\iBcolor}\footnotesize $\ell-1$} node[xshift = .15cm, yshift = .15cm,pos = 1] {\color{\rankcolor}\footnotesize $1$} (w2);

	    \draw (t1) edge node[xshift = -.1cm, yshift = .2cm,pos = 0] {\color{\rankcolor}\footnotesize $2$} node[xshift = -.2cm, yshift = -.2cm,pos = 1] {\color{\rankcolor}\footnotesize $2$} (b11);
	    \draw (t1) edge node[xshift = .07cm, yshift = .2cm,pos = 0] {\color{\rankcolor}\footnotesize $3$} node[xshift = -.15cm, yshift = .1cm,pos = 1] {\color{\rankcolor}\footnotesize $1$} (e12);
	    \draw (t1) edge node[xshift = .12cm, yshift = .1cm,pos = 0] {\color{\rankcolor}\footnotesize $4$} node[xshift = -.1cm, yshift = .2cm,pos = 1] {\color{\rankcolor}\footnotesize $1$} (b12);
	    \draw (t1) edge node[xshift = .15cm, yshift = -.1cm,pos = 0] {\color{\rankcolor}\footnotesize $5$} node[xshift = .1cm, yshift = .15cm,pos = 1] {\color{\rankcolor}\footnotesize $2$} (b13);

	    \draw (t1) edge node[xshift = -.15cm, yshift = -.15cm,pos = 0] {\color{\rankcolor}\footnotesize $1$} node[xshift = -.15cm, yshift = .15cm,pos = 1] {\color{\iAcolor}\footnotesize $2$} node[xshift = .15cm, yshift = .15cm,pos = 1] {\color{\iBcolor}\footnotesize $1$} (s1);
	    \draw (s1) edge node[xshift = -.15cm, yshift = -.15cm,pos = 0] {\color{\iAcolor}\footnotesize $1$} node[xshift = .15cm, yshift = -.15cm,pos = 0] {\color{\iBcolor}\footnotesize $2$} node[xshift = -.15cm, yshift = .15cm,pos = 1] {\color{\rankcolor}\footnotesize $2$} (t2);
	    \draw (t2) edge node[xshift = -.15cm, yshift = -.15cm,pos = 0] {\color{\rankcolor}\footnotesize $1$} node[xshift = -.15cm, yshift = .15cm,pos = 1] {\color{\rankcolor}\footnotesize $1$} (s2);
	    \draw (s2) edge node[xshift = -.15cm, yshift = -.15cm,pos = 0] {\color{\rankcolor}\footnotesize $2$} node[xshift = -.15cm, yshift = .15cm,pos = 1] {\color{\rankcolor}\footnotesize $1$} (t3);
      
	       \foreach \ang in {135,142,149,156,163}{
	       \draw (v2) edge[dashed] ($(v2) + (\ang:3*\edgedist)$);

	       }

	\end{tikzpicture}
	}

    \caption{Illustration of \Cref{thm:twooneside} for an instance $C_1\wedge C_2\wedge C_3$ with $C_1 = X_1\vee X_2 \vee X_3$, $C_2 = \overline{X}_4\vee \overline{X}_1\vee \overline{X}_2$, and $C_3 = X_5\vee X_4\vee X_1$. 
    Workers are represented by circles and firms by squares.
    For the sake of clarity, the edges between $v_2$ and vertices of $W\setminus \{ w_1,w_2,t_2,t_3\}$ are omitted and represented by the outgoing dashed edges from $v_2$. 
    These are best for agents of the type $\overline{d}_i^j$ and worst for the others. 
    The numbers on edges denote the rank in the preference lists, where $\ell$ and $\ell-1$ denoting the last second-to-last options, respectively.
    For the two agents $s_1$ and $v_2$ that changed their preferences in $\iB$ compared to $\iA$, we depict their preferences in $\iA$ on the left in blue and in $\iB$ on the right in red.
    }
    \label{fig:constrCombined}
\end{figure}

    For each clause $C_j$ ($j\in [m]$) and literal $\ell_{j_i}\in C_j$ ($i\in [3]$), we create a literal gadget $\gadg(\ell_{j_i})$:
    \begin{itemize}

        \item if $\ell_{j_i}$ is a positive literal, then $\gadg (\ell_{j_i})$ has workers $a_i^j,c_i^j,d_i^j\in W$ and firms $b_i^j,e_i^j,f_i^j\in F$.
        \item if $\ell_{j_i}$ is a negative literal, then $\gadg (\ell_{j_i})$ has workers $\nega_i^j,\negc_i^j,\negd_i^j\in W$ and firms $\negb_i^j,\nege_i^j,\negf_i^j\in F$.
    \end{itemize}

Moreover, we add workers $w_1,w_2$ and $t_1,t_2,t_3$ and firms $v_1,v_2$ and $s_1,s_2$. 

We denote $\overline{B} := \{\negb_i^j\colon \exists k\in [n] \text{ with } \ell_{j_i} = \overline{X}_k\}$ and $C := \{c_i^j\colon \exists k\in [n] \text{ with } \ell_{j_i} = X_k\}$.
Moreover, we denote $\overline{B}(c_i^j) := \{ \negb_{i'}^{j'}\colon \exists k\in [n] \text{ with } \ell_{j_i}=X_k \wedge \ell_{j'_{i'}}=\overline{X}_k\}$ 
(i.e., the set of those agents 
in $\overline{B}$ 
that correspond to literals which are exactly the negation of the $i$th literal in Clause $C_j$) 
and similarly, $C(\negb_i^j) := \{ c_{i'}^{j'}\colon \exists k\in [n] \text{ with } \ell_{j_i}=\overline{X}_k\wedge \ell_{j'_{i'}}=X_k\}$.
These sets will become important to define edges in the instance that are between the gadgets representing literals of different clauses.

Also, we define $A_3^j := \{ a_3^j\colon j\in [m]\}$, 
$\overline{A}_3^j := \{ \nega_3^j\colon j\in [m]\}$, 
$B_1^j := \{ b_1^j\colon j\in [m]\}$,
$\overline{B}_1^j := \{ \negb_1^j\colon j\in [m]\}$, 
and $\overline{E}_1^j := \{ \nege_1^j\mid j\in [m]\}$.
Finally, for a set $S$, let $[S]$ denote an arbitrary order of its elements.

\begin{table}
    \caption{Preferences in $\iA$ in the reduction of \Cref{thm:twooneside}.}
    \label{tab:twoonesideA}
\begin{center}
\renewcommand{\arraystretch}{1.3}
\begin{tabular}{cl|cl}
& Workers & & Firms \\\hline
   $a_i^j:$  & $f_i^j\succ b_{i+1}^j\succ v_2$  $(i\in [2])$ & $b_i^j:$ &  $c_i^j\succ a_{i-1}^j\succ d_i^j$ ($i=2,3)$\\
   $a_3^j:$  & $f_3^j\succ v_1\succ v_2$ & $b_1^j:$ & $c_1^j\succ t_1\succ d_1^j$ \\
   $\nega_i^j:$ &  $\negf_i^j\succ \negb_i^j\succ \nege_i^j\succ v_2$ ($i\in [2]$)& $\negb_i^j:$  & $\nega_{i-1}^j\succ \negc_i^j\succ [C(\negb_i^j)]\succ \negd_i^j$ ($i=2,3$)\\
   $\nega_3^j:$ & $\negf_3^j\succ v_1\succ v_2$ & $\negb_1^j:$ & $t_1\succ \negc_i^j\succ [C(\negb_1^j)]\succ \negd_i^j$ \\
$c_i^j:$ & $e_i^j\succ [\overline{B}(c_j^i)]\succ  b_i^j\succ v_2$ ($i\in [3])$ & $e_i^j:$ & $d_i^j\succ c_i^j$ ($i\in [3]$) \\
$\negc_i^j:$ & $\nege_i^j\succ \negb_i^j\succ v_2$ ($i=2,3)$ ($i\in [3]$) & $\nege_i^j:$ & $ \nega_{i-1}^j\succ \negd_i^j\succ \negc_i^j$ ($i = 2,3$)  \\

    $d_i^j:$ & $b_i^j\succ e_i^j\succ f_i^j\succ v_2$ ($i\in [3]$)&  $\nege_1^j:$ & $ t_1\succ \negd_i^j\succ \negc_i^j$  \\
$\negd_i^j:$  & $v_2\succ \negb_i^j\succ \negf_i^j\succ  \nege_i^j$ ($i\in [3]$) &  $f_i^j:$ & $d_i^j\succ a_i^j$ ($i\in [3]$)\\
 $t_1:$ & $s_1\succ [B_1^j\cup \overline{E}_1^j\cup \overline{B}_1^j]\succ v_2$  &  $\negf_i^j:$& $\nega_i^j\succ \negd_i^j$ ($i\in [3]$) \\
$t_2:$ & $s_2\succ s_1$ & $s_1:$ & $t_2\succ t_1$  \\
$t_3:$ & $s_2$ & $s_2:$ & $t_2\succ t_3$ \\
$w_1:$ & $v_2\succ v_1$ & $v_1:$ & $[A_3^j\cup \overline{A}_3^j]\succ w_1$ \\
 $w_2:$ & $v_2$ & $v_2:$ & $w_1\succ w_2\succ [W\setminus \{ t_2,t_3,w_1,w_2\}]$\\
\end{tabular}
\end{center}
\end{table}

The preferences of the agents in $\iA$ are stated in \Cref{tab:twoonesideA}.
In $\iB$, only the following two preference lists are different.
\begin{itemize}
    \item $s_1:$ $t_1\succ t_2$,
    \item $v_2:$ $[W\setminus \{ t_2,t_3,w_1,w_2\}]\succ w_2\succ w_1.$
\end{itemize}
Note that these agents are both firms, i.e., they are on the same side.

We now prove correctness of our reduction.
Specifically, we will prove equivalence of the following three statements:
\begin{enumerate}
    \item \label{it:sat} $\ins$ has a satisfying truth assignment.
    \item \label{it:dom} $\iP$ contains a robust dominant matching.
    \item \label{it:pop} $\iP$ contains a robust popular matching.
\end{enumerate}

Clearly, every robust dominant matching is robust popular, i.e., Statement~\ref{it:dom} implies Statement~\ref{it:pop}.
To complete the equivalence, we will show in the next two claims that Statement~\ref{it:sat} implies Statement~\ref{it:dom} and that Statement~\ref{it:pop} implies Statement~\ref{it:sat}.

\begin{claim}\label{claim:matching}
    Assume that $\ins$ has a satisfying truth assignment $\Phi$. Then, there exists a matching $M$ dominant in both $\iA$ and $\iB$.
\end{claim}
\begin{proof}
    We create a matching $M$ as follows. First, we add $\{ \{ s_i,t_i\} ,\{ v_i,w_i\} \mid i\in [2]\}$. 
    Then, for a positive literal $\ell_{j_i}$, if the literal is assigned True, we add $\{ \{b_i^j,c_i^j\}, \{ d_i^j,e_i^j\}, \{ a_i^j,f_i^j\} \}$, and if it is assigned False, then we add  $\{ \{b_i^j,d_i^j\}, \{ c_i^j,e_i^j\}, \{ a_i^j,f_i^j\} \}$.
    Moreover, for a negative literal $\ell_{j_i}$, if the literal is assigned True (and therefore corresponds to a False variable), we add $\{ \{ \negb_i^j,\negd_i^j\}, \{ \negc_i^j,\nege_i^j\}, \{ \nega_i^j,\negf_i^j\}\}$, and if it is assigned False, then we add $\{ \{ \negb_i^j,\negc_i^j\}, \{ \negd_i^j,\nege_i^j\}, \{ \nega_i^j,\negf_i^j\}\}$. See Figure~\ref{fig:matching}

    \begin{figure}
        \centering
		\newcommand{\rankcolor}{black}

		\resizebox{1\textwidth}{!}{
		\begin{tikzpicture}
		  \pgfmathsetmacro{\vertexsize}{0.6}
		  \pgfmathsetmacro{\gadgetrad}{1.2}
		  \pgfmathsetmacro{\edgedist}{.8}
		  \pgfmathsetmacro{\gadgetdist}{3.7}
		  \pgfmathsetmacro{\gadgetstretch}{7}

		  \foreach \k/\mult in {1/2} {
		    \pgfmathsetmacro{\y}{\mult * \gadgetdist}
		    \foreach \l/\x in {1/0,2/\gadgetstretch, 3/2*\gadgetstretch} {
		      \coordinate (p\l\k) at  (\x,\y) {};
		      \node[draw, rectangle, inner sep=0, minimum size=\vertexsize cm, minimum height=\vertexsize cm] (b\l\k) at ($(p\l\k)+(-\gadgetrad,0)$) {$b_{\l}^{\k}$};
		      \node[draw, rectangle, inner sep=0, minimum size=\vertexsize cm, minimum height=\vertexsize cm] (e\l\k) at ($(p\l\k)+(0,\gadgetrad)$) {$e_{\l}^{\k}$};
		      \node[draw, circle, inner sep=0, minimum size=\vertexsize cm, minimum height=\vertexsize cm] (c\l\k) at ($(p\l\k)+(0,-\gadgetrad)$) {$c_{\l}^{\k}$};
		      \node[draw, circle, inner sep=0, minimum size=\vertexsize cm, minimum height=\vertexsize cm] (d\l\k) at ($(p\l\k)+(\gadgetrad,0)$) {$d_{\l}^{\k}$};
		      \node[draw, rectangle, inner sep=0, minimum size=\vertexsize cm, minimum height=\vertexsize cm] (f\l\k) at ($(p\l\k)+(\gadgetstretch/2,0)+(-\edgedist,0)$) {$f_{\l}^{\k}$};
		      \node[draw, circle, inner sep=0, minimum size=\vertexsize cm, minimum height=\vertexsize cm] (a\l\k) at ($(p\l\k)+(\gadgetstretch/2,0)+(\edgedist,0)$) {$a_{\l}^{\k}$};

		      \draw (b\l\k) edge node[yshift = -.15cm,pos = .1] {\color{\rankcolor}\footnotesize $3$} node[yshift = -.15cm,pos = .95] {\color{\rankcolor}\footnotesize $1$} (d\l\k);
		      \draw (c\l\k) edge node[xshift =.15cm,pos = .1] {\color{\rankcolor}\footnotesize $1$} node[xshift =.15cm,pos = .9] {\color{\rankcolor}\footnotesize $2$} (e\l\k);
		      \draw (d\l\k) edge node[xshift =.1cm,yshift = .1cm, pos = .1] {\color{\rankcolor}\footnotesize $2$} node[xshift =.1cm,yshift = .1cm,pos = .9] {\color{\rankcolor}\footnotesize $1$} (e\l\k);
		      \draw (d\l\k) edge node[yshift = -.15cm,pos = .1] {\color{\rankcolor}\footnotesize $3$} node[yshift = -.15cm,pos = .9] {\color{\rankcolor}\footnotesize $1$} (f\l\k);
		      \draw (f\l\k) edge node[yshift = -.15cm,pos = .1] {\color{\rankcolor}\footnotesize $2$} node[yshift = -.15cm,pos = .92] {\color{\rankcolor}\footnotesize $1$} (a\l\k);
      
		    }
    
		      \draw (b1\k) edge node[below,pos = -.1] {\color{\rankcolor}\footnotesize $1$} node[below,pos = .7] {\color{\rankcolor}\footnotesize $3$} (c1\k);
		      \draw (b2\k) edge node[below,pos = -.1] {\color{\rankcolor}\footnotesize $1$} node[below,pos = .7] {\color{\rankcolor}\footnotesize $3$} (c2\k);
		      \draw (b3\k) edge node[below,pos = -.1] {\color{\rankcolor}\footnotesize $1$} node[below,pos = .7] {\color{\rankcolor}\footnotesize $2$} (c3\k);
		      \draw (a1\k) edge node[yshift = -.15cm,pos = .1] {\color{\rankcolor}\footnotesize $2$} node[yshift = -.15cm,pos = .85] {\color{\rankcolor}\footnotesize $2$} (b2\k);
		      \draw (a2\k) edge node[yshift = -.15cm,pos = .1] {\color{\rankcolor}\footnotesize $2$} node[yshift = -.15cm,pos = .85] {\color{\rankcolor}\footnotesize $2$} (b3\k);
    
		  }

		  \foreach \k/\mult in {2/1} {
		    \pgfmathsetmacro{\y}{\mult * \gadgetdist}
		    \foreach \l/\x in {1/0,2/\gadgetstretch, 3/2*\gadgetstretch} {
		      \coordinate (p\l\k) at  (\x,\y) {};
		      \node[draw, rectangle, inner sep=0, minimum size=\vertexsize cm, minimum height=\vertexsize cm] (b\l\k) at ($(p\l\k)+(225:\gadgetrad)$) {$\bar b_{\l}^{\k}$};
		      \node[draw, rectangle, inner sep=0, minimum size=\vertexsize cm, minimum height=\vertexsize cm] (e\l\k) at ($(p\l\k)+(135:\gadgetrad)$) {$\bar e_{\l}^{\k}$};
		      \node[draw, circle, inner sep=0, minimum size=\vertexsize cm, minimum height=\vertexsize cm] (c\l\k) at ($(p\l\k)+(45:\gadgetrad)$) {$\bar c_{\l}^{\k}$};
		      \node[draw, circle, inner sep=0, minimum size=\vertexsize cm, minimum height=\vertexsize cm] (d\l\k) at ($(p\l\k)+(315:\gadgetrad)$) {$\bar d_{\l}^{\k}$};
		      \node[draw, rectangle, inner sep=0, minimum size=\vertexsize cm, minimum height=\vertexsize cm] (f\l\k) at ($(p\l\k)+(\gadgetstretch/2,0)+(-\edgedist,0)$) {$\bar f_{\l}^{\k}$};
		      \node[draw, circle, inner sep=0, minimum size=\vertexsize cm, minimum height=\vertexsize cm] (a\l\k) at ($(p\l\k)+(\gadgetstretch/2,0)+(\edgedist,0)$) {$\bar a_{\l}^{\k}$};

		      \draw (b\l\k) edge node[xshift =-.1cm,yshift = .1cm,pos = .1] {\color{\rankcolor}\footnotesize $2$} node[xshift =-.1cm,yshift = .1cm,pos = .9] {\color{\rankcolor}\footnotesize $2$} (c\l\k);
		      \draw (c\l\k) edge node[yshift =.15cm,pos = .1] {\color{\rankcolor}\footnotesize $1$} node[yshift =.15cm,pos = .9] {\color{\rankcolor}\footnotesize $3$} (e\l\k);
		      \draw (d\l\k) edge node[xshift =.1cm,yshift = .1cm, pos = .1] {\color{\rankcolor}\footnotesize $4$} node[xshift =.1cm,yshift = .1cm,pos = .94] {\color{\rankcolor}\footnotesize $2$} (e\l\k);
		      \draw (d\l\k) edge node[yshift = -.15cm,pos = .1] {\color{\rankcolor}\footnotesize $3$} node[yshift = -.15cm,pos = .9] {\color{\rankcolor}\footnotesize $2$} (f\l\k);
		      \draw (f\l\k) edge node[yshift = -.15cm,pos = .1] {\color{\rankcolor}\footnotesize $1$} node[yshift = -.15cm,pos = .92] {\color{\rankcolor}\footnotesize $1$} (a\l\k);
      
		    }
		      \draw (b1\k) edge node[yshift = .15cm,pos = .1] {\color{\rankcolor}\footnotesize $4$} node[yshift = .15cm,pos = .95] {\color{\rankcolor}\footnotesize $2$} (d1\k);
		      \draw (b3\k) edge node[yshift = .15cm,pos = .1] {\color{\rankcolor}\footnotesize $4$} node[yshift = .15cm,pos = .95] {\color{\rankcolor}\footnotesize $2$} (d3\k);
		      \draw (b2\k) edge node[yshift = .15cm,pos = .1] {\color{\rankcolor}\footnotesize $5$} node[yshift = .15cm,pos = .95] {\color{\rankcolor}\footnotesize $2$} (d2\k);
		      \draw (a1\k) edge node[yshift = -.15cm,pos = .1] {\color{\rankcolor}\footnotesize $2$} node[yshift = -.15cm,pos = .92] {\color{\rankcolor}\footnotesize $1$} (b2\k);
		      \draw (a2\k) edge node[yshift = -.15cm,pos = .1] {\color{\rankcolor}\footnotesize $2$} node[yshift = -.15cm,pos = .92] {\color{\rankcolor}\footnotesize $1$} (b3\k);
		      \draw (a1\k) edge node[yshift = .15cm,pos = .05] {\color{\rankcolor}\footnotesize $3$} node[yshift = .15cm,pos = .92] {\color{\rankcolor}\footnotesize $1$} (e2\k);
		      \draw (a2\k) edge node[yshift = .15cm,pos = .05] {\color{\rankcolor}\footnotesize $3$} node[yshift = .15cm,pos = .92] {\color{\rankcolor}\footnotesize $1$} (e3\k);
		  }

		  \foreach \k/\mult in {3/0} {
		    \pgfmathsetmacro{\y}{\mult * \gadgetdist}
		    \foreach \l/\x in {1/0,2/\gadgetstretch, 3/2*\gadgetstretch} {
		      \coordinate (p\l\k) at  (\x,\y) {};
		      \node[draw, rectangle, inner sep=0, minimum size=\vertexsize cm, minimum height=\vertexsize cm] (b\l\k) at ($(p\l\k)+(-\gadgetrad,0)$) {$b_{\l}^{\k}$};
		      \node[draw, rectangle, inner sep=0, minimum size=\vertexsize cm, minimum height=\vertexsize cm] (e\l\k) at ($(p\l\k)+(0,-\gadgetrad)$) {$e_{\l}^{\k}$};
		      \node[draw, circle, inner sep=0, minimum size=\vertexsize cm, minimum height=\vertexsize cm] (c\l\k) at ($(p\l\k)+(0,\gadgetrad)$) {$c_{\l}^{\k}$};
		      \node[draw, circle, inner sep=0, minimum size=\vertexsize cm, minimum height=\vertexsize cm] (d\l\k) at ($(p\l\k)+(\gadgetrad,0)$) {$d_{\l}^{\k}$};
		      \node[draw, rectangle, inner sep=0, minimum size=\vertexsize cm, minimum height=\vertexsize cm] (f\l\k) at ($(p\l\k)+(\gadgetstretch/2,0)+(-\edgedist,0)$) {$f_{\l}^{\k}$};
		      \node[draw, circle, inner sep=0, minimum size=\vertexsize cm, minimum height=\vertexsize cm] (a\l\k) at ($(p\l\k)+(\gadgetstretch/2,0)+(\edgedist,0)$) {$a_{\l}^{\k}$};

		      \draw (b\l\k) edge node[yshift = .15cm,pos = .1] {\color{\rankcolor}\footnotesize $3$} node[yshift = .15cm,pos = .95] {\color{\rankcolor}\footnotesize $1$} (d\l\k);
		      \draw (c\l\k) edge node[xshift =.15cm,pos = .1] {\color{\rankcolor}\footnotesize $1$} node[xshift =.15cm,pos = .9] {\color{\rankcolor}\footnotesize $2$} (e\l\k);
		      \draw (d\l\k) edge node[yshift = -.2cm, pos = 0] {\color{\rankcolor}\footnotesize $2$} node[yshift = -.2cm,pos = .85] {\color{\rankcolor}\footnotesize $1$} (e\l\k);
		      \draw (d\l\k) edge node[yshift = -.15cm,pos = .1] {\color{\rankcolor}\footnotesize $3$} node[yshift = -.15cm,pos = .9] {\color{\rankcolor}\footnotesize $1$} (f\l\k);
		      \draw (f\l\k) edge node[yshift = -.15cm,pos = .1] {\color{\rankcolor}\footnotesize $2$} node[yshift = -.15cm,pos = .92] {\color{\rankcolor}\footnotesize $1$} (a\l\k);
      
		    }
    
		      \draw (b1\k) edge node[yshift = .2cm,pos = 0] {\color{\rankcolor}\footnotesize $1$} node[yshift = .2cm,pos = .75] {\color{\rankcolor}\footnotesize $2$} (c1\k);
		      \draw (b2\k) edge node[yshift = .2cm,pos = 0] {\color{\rankcolor}\footnotesize $1$} node[yshift = .2cm,pos = .75] {\color{\rankcolor}\footnotesize $3$} (c2\k);
		      \draw (b3\k) edge node[yshift = .2cm,pos = 0] {\color{\rankcolor}\footnotesize $1$} node[yshift = .2cm,pos = .75] {\color{\rankcolor}\footnotesize $3$} (c3\k);
		      \draw (a1\k) edge node[yshift = -.15cm,pos = .1] {\color{\rankcolor}\footnotesize $2$} node[yshift = -.15cm,pos = .85] {\color{\rankcolor}\footnotesize $2$} (b2\k);
		      \draw (a2\k) edge node[yshift = -.15cm,pos = .1] {\color{\rankcolor}\footnotesize $2$} node[yshift = -.15cm,pos = .85] {\color{\rankcolor}\footnotesize $2$} (b3\k);
    
		  }

		  \draw (c11) edge[bend left = 17] node[xshift = .1cm, yshift = .2cm,pos = 0] {\color{\rankcolor}\footnotesize $2$} node[yshift = .2cm,pos = 1] {\color{\rankcolor}\footnotesize $3$} (b22);
		  \draw (c21) edge[bend left = 17] node[xshift = .1cm, yshift = .2cm,pos = 0] {\color{\rankcolor}\footnotesize $2$} node[yshift = .2cm,pos = 1] {\color{\rankcolor}\footnotesize $3$} (b32);
  
		  \draw (b12) edge node[xshift = .1cm, yshift = -.2cm,pos = 0] {\color{\rankcolor}\footnotesize $3$} node[xshift = -.1cm, yshift = .2cm,pos = 1] {\color{\rankcolor}\footnotesize $2$} (c23);
		  \draw (b22) edge node[xshift = .1cm, yshift = -.2cm,pos = 0] {\color{\rankcolor}\footnotesize $4$} node[xshift = -.1cm, yshift = .2cm,pos = 1] {\color{\rankcolor}\footnotesize $2$} (c33);

		    \node[draw, rectangle, inner sep=0, minimum size=\vertexsize cm, minimum height=\vertexsize cm] (v1) at ($(a32)+(2*\edgedist,0)$) {$v_1$};
		    \node[draw, circle, inner sep=0, minimum size=\vertexsize cm, minimum height=\vertexsize cm] (w1) at ($(v1)+(0,-2*\edgedist)$) {$w_1$};
		    \node[draw, rectangle, inner sep=0, minimum size=\vertexsize cm, minimum height=\vertexsize cm] (v2) at ($(v1)+(0,-4*\edgedist)$) {$v_2$};
		    \node[draw, circle, inner sep=0, minimum size=\vertexsize cm, minimum height=\vertexsize cm] (w2) at ($(v1)+(0,-6*\edgedist)$) {$w_2$};  

		    \node[draw, circle, inner sep=0, minimum size=\vertexsize cm, minimum height=\vertexsize cm] (t1) at ($(b12)!.5!(e12)+(-2*\edgedist,0)$) {$t_1$};
		    \node[draw, rectangle, inner sep=0, minimum size=\vertexsize cm, minimum height=\vertexsize cm] (s1) at ($(t1)+(0,-2*\edgedist)$) {$s_1$};
		    \node[draw, circle, inner sep=0, minimum size=\vertexsize cm, minimum height=\vertexsize cm] (t2) at ($(t1)+(0,-4*\edgedist)$) {$t_2$};
		    \node[draw, rectangle, inner sep=0, minimum size=\vertexsize cm, minimum height=\vertexsize cm] (s2) at ($(t1)+(0,-6*\edgedist)$) {$s_2$};
		    \node[draw, circle, inner sep=0, minimum size=\vertexsize cm, minimum height=\vertexsize cm] (t3) at ($(t1)+(0,-8*\edgedist)$) {$t_3$};
     
		      \draw (a31) edge node[xshift = .2cm, yshift = -.1cm,pos = 0] {\color{\rankcolor}\footnotesize $2$} node[xshift = .05cm, yshift = .2cm,pos = 1] {\color{\rankcolor}\footnotesize $1$} (v1);
		      \draw (a32) edge node[xshift = .1cm, yshift = .15cm,pos = 0] {\color{\rankcolor}\footnotesize $2$} node[xshift = -.1cm, yshift = .15cm,pos = 1] {\color{\rankcolor}\footnotesize $2$} (v1);
		      \draw (a33) edge node[xshift = -.1cm, yshift = .2cm,pos = 0] {\color{\rankcolor}\footnotesize $2$} node[xshift = -.2cm, yshift = -.15cm,pos = 1] {\color{\rankcolor}\footnotesize $3$} (v1);

		       \draw (v1) edge node[xshift = .15cm, yshift = -.15cm,pos = 0] {\color{\rankcolor}\footnotesize $4$} node[xshift = .15cm, yshift = .15cm,pos = 1] {\color{\rankcolor}\footnotesize $2$} (w1);
		       \draw (w1) edge node[xshift = .15cm, yshift = -.15cm,pos = 0] {\color{\rankcolor}\footnotesize $1$} node[xshift = .15cm, yshift = .15cm,pos = 1] {\color{\rankcolor}\footnotesize $1$} (v2);
		       \draw (v2) edge node[xshift = .15cm, yshift = -.15cm,pos = 0] {\color{\rankcolor}\footnotesize $2$} node[xshift = .15cm, yshift = .15cm,pos = 1] {\color{\rankcolor}\footnotesize $1$} (w2);

		    \draw (t1) edge node[xshift = -.1cm, yshift = .2cm,pos = 0] {\color{\rankcolor}\footnotesize $2$} node[xshift = -.2cm, yshift = -.2cm,pos = 1] {\color{\rankcolor}\footnotesize $2$} (b11);
		    \draw (t1) edge node[xshift = .07cm, yshift = .2cm,pos = 0] {\color{\rankcolor}\footnotesize $3$} node[xshift = -.15cm, yshift = .1cm,pos = 1] {\color{\rankcolor}\footnotesize $1$} (e12);
		    \draw (t1) edge node[xshift = .12cm, yshift = .1cm,pos = 0] {\color{\rankcolor}\footnotesize $4$} node[xshift = -.1cm, yshift = .2cm,pos = 1] {\color{\rankcolor}\footnotesize $1$} (b12);
		    \draw (t1) edge node[xshift = .15cm, yshift = -.1cm,pos = 0] {\color{\rankcolor}\footnotesize $5$} node[xshift = .1cm, yshift = .15cm,pos = 1] {\color{\rankcolor}\footnotesize $2$} (b13);

		    \draw (t1) edge node[xshift = -.15cm, yshift = -.15cm,pos = 0] {\color{\rankcolor}\footnotesize $1$} node[xshift = -.15cm, yshift = .15cm,pos = 1] {\color{\rankcolor}\footnotesize $2$} (s1);
		    \draw (s1) edge node[xshift = -.15cm, yshift = -.15cm,pos = 0] {\color{\rankcolor}\footnotesize $1$} node[xshift = -.15cm, yshift = .15cm,pos = 1] {\color{\rankcolor}\footnotesize $2$} (t2);
		    \draw (t2) edge node[xshift = -.15cm, yshift = -.15cm,pos = 0] {\color{\rankcolor}\footnotesize $1$} node[xshift = -.15cm, yshift = .15cm,pos = 1] {\color{\rankcolor}\footnotesize $1$} (s2);
		    \draw (s2) edge node[xshift = -.15cm, yshift = -.15cm,pos = 0] {\color{\rankcolor}\footnotesize $2$} node[xshift = -.15cm, yshift = .15cm,pos = 1] {\color{\rankcolor}\footnotesize $1$} (t3);
      
		       \foreach \ang in {135,142,149,156,163}{
		       \draw (v2) edge[dashed] ($(v2) + (\ang:3*\edgedist)$);

		       }

		    \foreach \v/\w in {t1/s1,t2/s2,v1/w1,v2/w2,b11/d11,c11/e11,b21/c21,d21/e21,b31/c31,d31/e31,e12/c12,b12/d12,e22/c22,b22/d22,e32/d32,b32/c32,b13/c13,d13/e13,b23/d23,c23/e23,b33/d33,c33/e33}
		    \draw (\v) edge[very thick] (\w);

		    \foreach \i in {1,2,3}
		    {
		    \foreach \j in {1,2,3}
		    {
		    \draw (f\i\j) edge[very thick] (a\i\j);
		    }
		    }
  
		\end{tikzpicture}
		} 
        
        \caption{Illustration of the constructed matching $M$ in \Cref{claim:matching} for the instance $\iA$ in the example of \Cref{fig:constrCombined}.
        The matching corresponds to the assignment where $X_1$ and $X_4$ are set to False, and $X_2$, $X_3$, and $X_5$ are set to True.}
        \label{fig:matching}
    \end{figure}

    Clearly, $M$ covers all vertices except $t_3$, hence $M$ is a maximum size matching. As a consequence, $M$ is popular in either instance if and only if it is dominant.
    Therefore, it suffices to show that $M$ is popular in both $\iA$ and $\iB$.

    First, we focus on $\iA$. 
    We want to use the characterization of \Cref{thm:pop_char} to prove popularity.
    Consider $G_M^{\iA}$ and the $(-,-)$ (i.e., blocking) edges. 
    As $\Phi$ is consistent (i.e., there is no $i\in [n]$, such that both $X_i$ and $\overline{X}_i$ is set to True somewhere), none of the edges between $\overline{B}=\{ \negb_i^j\mid j\in [m], i\in [3]\}$ and $C=\{ c_i^j\mid j\in [m], i\in [3]\}$ are a blocking edge, as either $c_i^j$ or $b_i^j$ obtains a better partner. 
    Hence, it is straightforward to verify that the only $(-,-)$ edge is $\{ v_2,w_1\}$.

    Therefore, it is enough to show that in $G_M^{\iA}$, there is no alternating path with respect to $M$ between $\{ v_2,w_1\}$ and $t_3$ and no alternating cycle through $\{ v_2,w_1\}$. The latter is impossible, as $\{ v_2,w_2\}\in M$ and $w_2$ is a leaf. 
    
    Now, suppose for the contrary that there exists an alternating path with respect to $M$ between $\{ v_2,w_1\}$ and $t_3$, and denote it by $P$.
    Note that $P$ cannot use edges containing $v_2$ other than $\{v_2,w_1\}$.
    Indeed, $\{v_2,w_2\}\in M$ and $w_2$ is a leaf, so the path cannot enter from $w_2$ and would get stuck when entering via a nonmatching edge to $v_2$.
    
    We now claim that $t_3$ is in a different component in $G_M^{\iA}\setminus \{ \{ v_2,w\} \colon w\in W,w\neq t_3\}$ than $v_2$, 
    which contradicts the existence of such an alternating path.
    
    As all clauses $C_j$ contain a True literal, for all $j$, at least one edge in the chain of literal gadgets $\gadg (\ell_{j_1}),\gadg (\ell_{j_2}),\gadg (\ell_{j_3})$ is a $(+,+)$ edge: if these are negative literals, then it is $\{ \negd_i^j,\negf_i^j\}$ for a true literal $\ell_{j_i}$ and if these are positive literals, then it is $\{ a_{i-1}^j,b_i^j\}$ for a true literal $\ell_{j_i}$ (there, we identify $a_0^j:=t_1)$. Furthermore, it is easy to see that $\{ \negd_i^j,\negf_i^j\}$ and $\{ a_{i-1}^j,b_i^j\}$ are both unavoidable for any alternating path between $\{ v_2,w_1\}$ and $t_3$ that does not use edges between $\overline{B}$  and $C$. Hence, $P$ must contain an edge between $\overline{B}$ and $C,$ as it cannot use $(+,+)$ edges.

    We now show that this leads to a contradiction.
    If we orient the path $P$ towards $t_3$, then it is clear that any edge of $P\setminus M$ is directed from firms towards workers.
    Furthermore, any edge between $\overline{B}$ and $C$ is either a $(+,+)$ edge, or a $(-,+)$ edge. 
    Suppose such a $(-,+)$ edge $\{ c_i^j,\negb_{i'}^{j'}\}$ is used in the path and take the one closest to $v_2$.
    Since $\overline{B}$ only contains firms and $C$ only workers, it must be the firm $\negb_{i'}^{j'}$ that is closer to $v_2$ on the path. 
    Since the path reached $\negb_{i'}^{j'}$ without edges from $\overline{B}$ to $C$, the edge $\{ \negd_{i'}^{j'},\negf_{i'}^{j'}\}$ is not a $(+,+)$ edge, so $\{ \negd_{i'}^{j'},\nege_{i'}^{j'}\}, \{ \negb_{i'}^{j'},\negc_{i'}^{j'}\}\in M$. 
    Moreover, as $\{ c_i^j,\negb_{i'}^{j'}\}$ is a $(-,+)$ edge, we must have $\{ b_i^j,c_i^j\}, \{ d_i^j,e_i^j\} \in M$.
    If $\{ b_i^j,d_i^j\} \in P$, then $ \{ d_i^j,e_i^j\} \in P$ and $e_i^j$ must be the end of $P$ (as the only neighbors of $e_i^j$ are $c_i^j$ and $d_i^j$ who are already on the path). This is a contradiction. 
    However, the other edge adjacent to $b_i^j$ (either $\{ b_i^j,a_{i-1}\}$ or $\{ b_i^j,t_1\}$ is a $(+,+)$ edge in $\iA$), and we derive a contradiction once again.

    Therefore, we have shown that there can be no alternating path with respect to $M$ between $\{ v_2,w_1\}$ and $t_3$ in $G_M^{\iA}$, as it would need to use edges between $\overline{B}$ and $C$ as we have shown first, but that is impossible, as we have shown above. So $M$ is popular in $\iA$.

    Next, we consider $\iB$. 
    Here, $\{ s_1,t_2\}$ is a $(+,+)$ edge, hence $t_3$ is isolated in $G_M^{\iB}$. Similarly as before, no edge between $\overline{B}$ and $C$ is a $(-,-)$ edge. 
    Hence, the only $(-,-)$ edges are all incident to $v_2$. 
    Again, as $\{ v_2,w_2\} \in M$, there can be no alternating cycle through any of them. 
    Moreover, as $t_3$ is isolated, there cannot be any alternating path in $G_M^{\iB}$ between $(-,-)$ edges, or between a $(-,-)$ edge and an uncovered vertex by $M$. Hence, $M$ is popular in $\iB$, too.
\end{proof}

It remains to show that popular matchings imply the existence of satisfying assignments.

\begin{claim}
    If there is a matching $M$ that is popular in both $\iA$ and $\iB$, then there exists a satisfying truth assignment $\Phi$ to $\varphi$.
\end{claim}
\begin{proof}
    Let $M$ be such a matching. First, we show that $M$ must match all agents except for $t_3$. 
    
    On the one hand, $\{ v_2,w_2\}\in M$, since otherwise $w_2$ in unmatched by $M$, which is a contradiction as shown next.
    Indeed, there is $X\in \{A,B\}$ such that $w_2\succ_{v_2}^{\iX}M(v_2)$, so $\{ v_2,w_2\}$ itself constitutes an alternating path in $G_M^{\iX}$ between an unmatched agent that contains a $(-,-)$ edge, contradicting that $M$ is popular in $\iX$.
    On the other hand, suppose now that some $w\in W\setminus \{ t_2,t_3\}$ with $w\neq w_2$ is not matched. 
    Then again, there is $X\in \{A,B\}$ such that $w\succ_{v_2}^{\iX}w_2$, which is a contradiction for the same reason. Finally, $t_2$ must be matched, otherwise, in both instances, $\{s_2,t_2\}$ is a $(-,-)$ edge adjacent to an unmatched vertex, contradicting popularity. 
    Hence, since there is precisely one more worker than firms and all workers except $t_3$ are matched, $M$ must match all agents except for $t_3$.

    This immediately implies that $\{ v_1,w_1\},\{ v_2,w_2\},\{ s_1,t_1\},\{ s_2,t_2\} \in M$. 
    We now claim that no edge between $\overline{B}$ and $C$ is included in $M$.
    Assume for contradiction that, for some $j\in [m]$ and $i\in [3]$, $c_i^j$ was matched by such an edge.
    Then, as $e_i^j$ has only two neighbors, we have $\{e_i^j,d_i^j\}\in M$.  Continuing this, we iteratively find firms  on the literal gadget for which all but one possible neighbors have already been matched, enforcing further matching edges.
    If $i = 1$, we would need that $\{b_i^j,t_1\}\in M$, which is impossible since $\{ s_1,t_1\}$.
    Hence, $i\ge 2$ and $\{b_i^j,a_{i-1}^j\}\in M$.
    We further have $\{f_{i-1},d_{i-1}^j\}, \{e_{i-1}^j,c_{i-1}^j\}\in M$.
    Now, if $i = 2$, we would need that $\{b_{i-1}^j,t_1\}\in M$, which is again impossible.
    Hence, we must have $i = 3$.
    Tracing further, we need $\{b_2^j,a_1^j\},\{f_1,d_1^j\}, \{e_1^j,c_1^j\}\in M$.
    Hence, $b_1^j$ must be matched to their only remaining neighbor, $t_1$, a contradiction.
    We conclude that no edge between $\overline{B}$ and $C$ is included in $M$.
    
    Next, we consider the chains of negative literal gadgets.
    We claim that $\{ \nega_i^j,\negf_i^j\}\in M$ for $j\in [m]$, $i\in [3]$. 
    Indeed, assume that this was not the case and consider an edge $\{ \nega_i^j,\negf_i^j\}\notin M$ such that $\{ \nega_{i'}^{j},\negf_{i'}^{j}\}\in M$ for all $i'\in [3]$ with $i' < i$, i.e., $i$ is minimal in $[3]$ with this property.
    Note that $\{ \negd_i^j,\negf_i^j\}\in M$ to match $\negf_i^j$.
    Now, since no edge between $\overline{B}$ and $C$ is included in $M$ and by the minimality of $j$ the agents in $\{\negc_i^j,\negd_i^j,\nege_i^j\}$ have to be matched among themselves. 
    Hence, one of them remains unmatched, a contradiction.
    It follows that $\{ \nega_i^j,\negf_i^j\}\in M$ for $j\in [m]$, $i\in [3]$.
    To match all agents in the negative literal gadget, it must hold that 
    for any $j\in [m]$, $i\in [3]$, either $\{ \{ \negb_i^j,\negd_i^j\}, \{ \negc_i^j,\nege_i^j\} , \{ \nega_i^j,\negf_i^j\} \} \subseteq M$ or $\{ \{ \negb_i^j,\negc_i^j\}, \{ \negd_i^j,\nege_i^j\} , \{ \nega_i^j,\negf_i^j\} \} \subseteq M$. 
    Then, consider the chains of the positive literal gadgets. 
    Applying a similar 
    argument, we get that for any $j\in [m]$, $i\in [3]$, either $\{  \{ b_i^j,c_i^j\}, \{ d_i^j,e_i^j\}, \{ a_i^j,f_i^j\}\} \subseteq M$ or $\{  \{ b_i^j,d_i^j\}, \{ c_i^j,e_i^j\}, \{ a_i^j,f_i^j\}\} \subseteq M$.

    Define a truth assignment $\Phi$ as follows. 
    We set variable $X_l$ as True, if there exists a positive literal $\ell_{j_i}=X_l$ such that $\{  \{ b_i^j,c_i^j\}, \{ d_i^j,e_i^j\}, \{ a_i^j,f_i^j\}\} \subseteq M$ and False if there is a negative literal $\ell_{j'_{i'}}=\overline{X}_l$ such that $\{ \{ \negb_{i'}^{j'},\negd_{i'}^{j'}\}, \{ \negc_{i'}^{j'},\nege_{i'}^{j'}\} , \{ \nega_{i'}^{j'},\negf_{i'}^{j'}\} \} \subseteq M$. 
    If neither holds, then we also set $X_l$ as True.

    First of all, we show that this is a consistent truth assignment. 
    Suppose for the contrary, that there is a variable $X_l$, such that there exists literals $\ell_{j_i}=X_l$ and $\ell_{j'_{i'}}=\overline{X}_l$, such that $\{  \{ b_i^j,c_i^j\}, \{ d_i^j,e_i^j\}, \{ a_i^j,f_i^j\}\} \subseteq M$ and $\{ \{ \negb_{i'}^{j'},\negd_{i'}^{j'}\}, \{ \negc_{i'}^{j'},\nege_{i'}^{j'}\} , \{ \nega_{i'}^{j'},\negf_{i'}^{j'}\} \} \subseteq M$ both hold. Then, $\{ c_i^j,\negb_{i'}^{j'}\}$ is a $(-,-)$ edge in both $\iA$ and $\iB$ and hence the alternating path $\{ \{c_i^j,\negb_{i'}^{j'}\}, \{ \negb_{i'}^{j'},\negd_{i'}^{j'}\} ,\{ \negd_{i'}^{j'},v_2\} \}$ is a subset of $G_M^{\iB}$ and witnesses the unpopularity of $M$ in $\iB$, a contradiction.  

    Finally, we show that $\Phi$ satisfies every clause. Suppose that a positive clause $C_j=(\ell_{j_1}\vee \ell_{j_2}\vee \ell_{j_3})$ (containing only positive literals) only contains literals set to False. Then, in $G_M^{\iA}$, the alternating path $\{ \{ v_2,w_1\}, \{ w_1,v_1\}, \{ v_1, a_3^j\} \} \cup \{ \{ a_i^j,f_i^j\}, \{ f_i^j,d_i^j\}, \{ d_i^j,b_i^j\}, \{ b_i^j,a_{i-1}^j\} \mid i\in [3]\} \cup \{ \{ t_1,s_1\}, \{ s_1,t_2\},\{ t_2,s_2\},\{s_2,t_3\}\}$ (here we let $a_0^j:=t_1$) between the uncovered agent $t_3$ and the blocking edge $\{ v_2,w_1\}$ witnesses the unpopularity of $M$ in $\iA$, a contradiction.

    Suppose that a negative clause only contains literals set to True. 
    Then, in $G_M^{\iA}$, the alternating path $\{ \{ v_2,w_1\}, \{ w_1,v_1\}, \{ v_1, \nega_3^j\} \} \cup \{ \{ \nega_i^j,\negf_i^j\}, \{ \negf_i^j,\negd_i^j\}, \{ \negd_i^j,\nege_i^j\}, \{ \nege_i^j,\nega_{i-1}^j\} \mid i\in [3]\} \cup \{ \{ t_1,s_1\}, \{ s_1,t_2\},\{ t_2,s_2\},\{s_2,t_3\}\}$ (here we let $\nega_0^j:=t_1$) between the uncovered agent $t_3$ and the blocking edge $\{ v_2,w_1\}$ witnesses the unpopularity of $M$ in $\iA$, a contradiction.
    We conclude that $\Phi$ is a satisfying truth assignment to $\varphi$.
\end{proof}

Since, we have shown the equivalence of all three statements, we have shown 
\NP-hardness of both \RobustProb\ and \RobustDom. 
\end{proof}

\subsubsection{Perturbations by Two Swaps}

Next, we consider our hardness results when two agents (but not from the same side) only perform a swap.
Our proof idea is to reduce from \ForbEdgeForceVert, i.e., the problem of finding a popular matching that avoids a designated edge while matching a designated agent.
In the reduced instance of \RobustProb, we use one of the instances to represent the input instance and to contain all originally popular matchings.
In the second instance, we perform two swaps that prevent the designated edge and ensure matching of the designated vertex.
We achieve this by introducing auxiliary agents.
For the designated edge they represent the case of the leaf agent being unmatched, and for the designated vertex they present an outside option to ensure them being matched in every popular matching.
However, the swap of preferences prevents matching with the outside option in a popular matching.

\begin{theorem}\label{thm:hardness}
    \RobustProb is \NP-complete even if the perturbed instance only differs by a swap of two agents.
\end{theorem}
\begin{proof}
    First, note that membership in {\NP} is straightforward.
    A robust popular matching with respect to two given input instances of {\MatP} serves as a polynomial-size certificate for a Yes-instance.
    We can verify it by simply checking whether the matching is popular in both instances in polynomial time \citep[Theorem~9]{BIM10a}.

    For \NP-hardness, we perform a reduction from \ForbEdgeForceVert, as defined in \Cref{sec:problems}.
    For this, consider an instance of \ForbEdgeForceVert given by an instance $\ins$ of \MatP on the graph $(W\cup F, E)$ together with a designated edge $e = \{a,b\}\in E$ and a vertex $d\in W\cup F$, where $N_a^{\ins} = \{b\}$, $b$ only contains one neighbor $c$ other than $a$, and $b$ and $c$ top-rank each other.
    
    The idea of the reduced instance $\iP$ is to enhance the source instance by a set of auxiliary agents for the designated edge and vertex.
    An illustration of the key agents is given in \Cref{fig:forbiddenedge,fig:forcedvertex}.
    Let $N' = W\cup F \cup \{\ell_a, r_a, \ell_d,r_d\}$ be the agent set. 
    Moreover, the edge set is given by $E' = E \cup \{\{a,\ell_a\},\{r_a,\ell_a\},\{d,\ell_d\},\{r_d,\ell_d\}\}$.
    Hence, the underlying graph is still bipartite, but whether the auxiliary agents are workers or firms, depends on the agent types of $a$ and $d$.

    For $x\in \{a,d\}$, we define $S_x = \{x, \ell_x, r_x\}$, i.e., the set containing this agent and their associated auxiliary agents.
    
    \begin{figure}
        \centering
        \begin{tikzpicture}
            \pgfmathsetmacro\vertexsize{.85}
            \pgfmathsetmacro\halfspan{1.1}
            \pgfmathsetmacro\graphspan{2.2}
            \pgfmathsetmacro\graphheight{1.2}
            \node[draw, circle, minimum size = \vertexsize cm, label = 180:{\color{myred}$b\succ^{\iA}_a \ell_a$}] (a) at (0,0) {$a$};
            \node[draw, circle, minimum size = \vertexsize cm, label = 0:$c\succ^{\iA}_b a$] (b) at (\graphspan,0) {$b$};
            \node[draw, circle, minimum size = \vertexsize cm] (ra) at (0,\graphheight) {$r_a$};
            \node[draw, circle, minimum size = \vertexsize cm, label = 0:$a \succ^{\iA}_{\ell_a} r_a$] (la) at (\graphspan,\graphheight) {$\ell_a$};
            \node[draw, circle, minimum size = \vertexsize cm, label = 180:$b \succ^{\iA}_{c} \dots$] (lb) at (0,-\graphheight) {$c$};

            \draw[thick] (ra) -- (la) -- (a) -- (b) -- (lb);
            \draw[thick] (lb) -- ($(lb)+ (0:\halfspan)$);
            \draw[thick,dotted] ($(lb)+ (0:\halfspan)$) -- ($(lb)+ (0:\halfspan) + (0:.35)$);
            \draw[thick] (lb) -- ($(lb)+ (330:\halfspan)$);
            \draw[thick,dotted] ($(lb)+ (330:\halfspan)$) -- ($(lb)+ (330:\halfspan) + (330:.35)$);
        \end{tikzpicture}
        \quad
        \begin{tikzpicture}
            \pgfmathsetmacro\vertexsize{.85}
            \pgfmathsetmacro\halfspan{1.1}
            \pgfmathsetmacro\graphspan{2.2}
            \pgfmathsetmacro\graphheight{1.2}
            \node[draw, circle, minimum size = \vertexsize cm, label = 180:{\color{myred}$\ell_a\succ^{\iB}_a b$}] (a) at (0,0) {$a$};
            \node[draw, circle, minimum size = \vertexsize cm, label = 0:$c\succ^{\iB}_b a$] (b) at (\graphspan,0) {$b$};
            \node[draw, circle, minimum size = \vertexsize cm] (ra) at (0,\graphheight) {$r_a$};
            \node[draw, circle, minimum size = \vertexsize cm, label = 0:$a \succ^{\iB}_{\ell_a} r_a$] (la) at (\graphspan,\graphheight) {$\ell_a$};
            \node[draw, circle, minimum size = \vertexsize cm, label = 180:$b \succ^{\iB}_{c} \dots$] (lb) at (0,-\graphheight) {$c$};

            \draw[thick] (ra) -- (la) -- (a) -- (b) -- (lb);
            \draw[thick] (lb) -- ($(lb)+ (0:\halfspan)$);
            \draw[thick,dotted] ($(lb)+ (0:\halfspan)$) -- ($(lb)+ (0:\halfspan) + (0:.35)$);
            \draw[thick] (lb) -- ($(lb)+ (330:\halfspan)$);
            \draw[thick,dotted] ($(lb)+ (330:\halfspan)$) -- ($(lb)+ (330:\halfspan) + (330:.35)$);
        \end{tikzpicture}
        \caption{Auxiliary agents for the forbidden edge $e = \{a,b\}$ of the reduced instances in the proof of \Cref{thm:hardness}. The preferences in $\iA$ and $\iB$ are described in the left and right picture, respectively.
        The only agent that changes their preferences by a simple swap is agent~$a$, as highlighted in red.
        }
        \label{fig:forbiddenedge}
    \end{figure}
    
    \begin{figure}
        \centering
        \begin{tikzpicture}
            \pgfmathsetmacro\vertexsize{.85}
            \pgfmathsetmacro\halfspan{1.1}
            \pgfmathsetmacro\graphspan{2.2}
            \pgfmathsetmacro\graphheight{1.2}
            \node[draw, circle, minimum size = \vertexsize cm, label = 180:{$\dots \succ^{\iA}_d \ell_d$}] (a) at (0,0) {$d$};
            \node[draw, circle, minimum size = \vertexsize cm] (ra) at (0,\graphheight) {$r_d$};
            \node[draw, circle, minimum size = \vertexsize cm, label = 0:{\color{myred}$d \succ^{\iA}_{\ell_d} r_d$}] (la) at (\graphspan,\graphheight) {$\ell_d$};

            \draw[thick] (ra) -- (la) -- (a);
            \draw[thick] (a) -- ($(a)+ (0:\halfspan)$);
            \draw[thick,dotted] ($(a)+ (0:\halfspan)$) -- ($(a)+ (0:\halfspan) + (0:.35)$);
            \draw[thick] (a) -- ($(a)+ (330:\halfspan)$);
            \draw[thick,dotted] ($(a)+ (330:\halfspan)$) -- ($(a)+ (330:\halfspan) + (330:.35)$);
        \end{tikzpicture}
        \quad
        \begin{tikzpicture}
            \pgfmathsetmacro\vertexsize{.85}
            \pgfmathsetmacro\halfspan{1.1}
            \pgfmathsetmacro\graphspan{2.2}
            \pgfmathsetmacro\graphheight{1.2}
            \node[draw, circle, minimum size = \vertexsize cm, label = 180:{$\dots \succ^{\iB}_d \ell_d$}] (a) at (0,0) {$d$};
            \node[draw, circle, minimum size = \vertexsize cm] (ra) at (0,\graphheight) {$r_d$};
            \node[draw, circle, minimum size = \vertexsize cm, label = 0:{\color{myred}$r_d \succ^{\iB}_{\ell_d} d$}] (la) at (\graphspan,\graphheight) {$\ell_d$};

            \draw[thick] (ra) -- (la) -- (a);
            \draw[thick] (a) -- ($(a)+ (0:\halfspan)$);
            \draw[thick,dotted] ($(a)+ (0:\halfspan)$) -- ($(a)+ (0:\halfspan) + (0:.35)$);
            \draw[thick] (a) -- ($(a)+ (330:\halfspan)$);
            \draw[thick,dotted] ($(a)+ (330:\halfspan)$) -- ($(a)+ (330:\halfspan) + (330:.35)$);
        \end{tikzpicture}
        \caption{Auxiliary agents for the forced vertex $d$ of the reduced instances in the proof of \Cref{thm:hardness}. The preferences in $\iA$ and $\iB$ are described in the left and right picture, respectively.
        The only agent that changes their preferences by a simple swap is agent~$\ell_d$, as highlighted in red.
        }
        \label{fig:forcedvertex}
    \end{figure}

    The preferences are mostly inherited from the source instance.
    For $i\in (W\cup F)\setminus \{a,d\}$, we define $\succ_i^{\iA} = \succ_i^{\ins}$ and $\succ_i^{\iB} = \succ_i^{\ins}$.
    Moreover, for both instances, $X\in \{A,B\}$,
    \begin{itemize}
        \item the preference list of $\ell_a$ is $a \succ^{\iX}_{\ell_a} r_a$,
        \item the preference list of $d$ appends $\ell_d$ at the very end, i.e., for all $x,x'\in N_x^{\ins}$, we have $x \succ^{\iX}_{d} \ell_d$ as well as $x \succ^{\iX}_{d} x'$ if and only if $x \succ^{\ins}_{d} x'$, and
        \item the preference lists of $r_a$ and $r_d$ only contain their single neighbor.
    \end{itemize}

    The only difference in the preferences are for agents $a$ and $\ell_d$.
    Specifically, we have $b\succ^{\iA}_a \ell_a$, $\ell_a \succ^{\iB}_a b$, $d \succ^{\iA}_{\ell_d} r_d$, and $r_d \succ^{\iB}_{\ell_d} d$.
    Note that both preference orders differ in exactly a swap (of their only two neighbors).
    
    We are ready to prove the correctness of the reduction.
    To this end, we will show that $\ins$ contains a popular matching $M$ that covers $d$ and with $e\notin M$ if and only if the reduced instance contains a matching popular for both $\iA$ and $\iB$.

    $\implies$
    Assume first that $\ins$ contains a popular matching $M$ that covers $d$ and with $e\notin M$.
    Define the matching $M' = M \cup \{\{a,\ell_a\},\{r_d,\ell_d\}\}$.
    We claim that $M'$ is popular for both $\iA$ and $\iB$.
    
    Assume for contradiction that $M'$ is not popular for instance $\iA$ and that there exists a matching $\hat M'$ with $\pmarg^{\iA}(\hat M',M') > 0$.

    Define $\hat M = \{e\in \hat M'\colon e\subseteq W\cup F\}$, i.e., the matching $\hat M'$ restricted to agents present in the source instance.
    We will argue that $\pmarg^{\ins}(\hat M,M) > 0$.
    Let $i\in (W\cup F)\setminus \{a,d\}$.
    Then, it holds that $\hat M(i) = \hat M'(i)$ if $i$ is matched in $\hat M'$ or $i$ is unmatched in both.
    Since the preferences of $i$ are identical in $\ins$ and $\iA$, we have
    that $\vote_i^{\ins}(\hat M,M) = \vote_i^{\iA}(\hat M',M')$.
    We refer to this as Observation~$(\diamond)$.

    Next, consider agent $a$.
    We will now show that $\vote_a^{\ins}(\hat M,M) \ge \vote_{S_a}^{\iA}(\hat M',M')$, and refer to this fact as Observation~$(\diamond\diamond)$.

    First, note that, whenever $r_a$ votes in favor of $\hat M'$, then $\ell_a$ votes in favor of $M'$ and hence $\vote_{\{r_a,\ell_a\}}^{\iA}(\hat M',M') \le 0$.
    Now, recall that $a$ is unmatched in $M$ because $e\notin M$. 
    Then, if $a$ is unmatched in $\hat M$, then $\vote_a^{\ins}(\hat M,M) = 0$ while $\vote_a^{\iA}(\hat M',M') \le 0$ Indeed, either $a$ is still matched in $a$ or is unmatched, in which case they prefer $M'$.
    Hence, $\vote_a^{\ins}(\hat M,M) = 0 \ge \vote_a^{\iA}(\hat M',M') + \vote_{\{r_a,\ell_a\}}^{\iA}(\hat M',M') = \vote_{S_a}^{\iA}(\hat M',M')$, as desired.

    Otherwise, $a$ is matched in $\hat M$, i.e., $\{a,b\}\in \hat M$.
    Then, since $\hat M$ is the restriction of $\hat M'$, we have that $\{a,b\}\in \hat M'$, and hence $\vote_a^{\ins}(\hat M,M) = \vote_a^{\ins}(\hat M',M')$.
    Hence, we again have that $\vote_a^{\ins}(\hat M,M) \ge \vote_{S_a}^{\iA}(\hat M',M')$.
    This concludes the proof of Observation~$(\diamond\diamond)$.
    
    Next, consider agent $d$.
    We will now show that $\vote_d^{\ins}(\hat M,M) \ge \vote_{S_d}^{\iA}(\hat M',M')$, and refer to this fact as Observation~$(\diamond\diamond\diamond)$.

    First, note that, whenever $\ell_d$ votes in favor of $\hat M'$, then $r_d$ will be unmatched in $\hat M'$ and, therefore, votes in favor of $M'$.
    Hence, $\vote_{\{r_d,\ell_d\}}^{\iA}(\hat M',M') \le 0$.
    
    Now, recall that $d$ is matched in $M$.
    If $d$ is matched in $\hat M$, then they perform the same vote and we have $\vote_d^{\ins}(\hat M,M) = \vote_d^{\ins}(\hat M',M')$.
    Hence, $\vote_d^{\ins}(\hat M,M) \ge \vote_{S_d}^{\iA}(\hat M',M')$.
    Otherwise, $d$ is unmatched in $\hat M$. 
    Hence, they are unmatched in $\hat M'$ or $\hat M'(d) = \ell_d$, both of which is worse than $M'(d)$, because $\ell_d$ is last-ranked by $d$.
    Hence, $\vote_d^{\ins}(\hat M,M) = \vote_d^{\ins}(\hat M',M') = -1$, and we again conclude that $\vote_d^{\ins}(\hat M,M) \ge \vote_{S_d}^{\iA}(\hat M',M')$.
    This concludes the proof of Observation~$(\diamond\diamond\diamond)$.
    
    Combining Observations~$(\diamond)$, $(\diamond\diamond)$, and~$(\diamond\diamond\diamond)$, we obtain
    \begin{align*}
        \pmarg^{\ins}(\hat M, M) &= \sum_{x\in (W\cup F)\setminus \{a,d\}} \vote_x^{\ins}(\hat M, M) + \sum_{x\in \{a,d\}} \vote_x^{\ins}(\hat M, M)\\
        & \ge \sum_{x\in (W\cup F)\setminus \{a,d\}} \vote_x^{\iA}(\hat M', M') + \sum_{x\in \{a,d\}} \vote_{S_x}^{\iA}(\hat M', M')\\
        &= \pmarg^{\iA}(\hat M', M') > 0\text.
    \end{align*}
    This contradicts the popularity of $M$.
    Hence, we have derived a contradiction and $M'$ is popular for $\iA$.

    Now, since $\iB$ only differs from $\iA$ by a swap of agents that are not matching partners in $M'$, we have that for every matching $\hat M'$ and every agent $x\in W'\cup F'$, it holds that $\vote_x^{\iB}(M', \hat M') \ge \vote_x^{\iA}(M',\hat M')$.
    Therefore, $\pmarg^{\iB}(M',\hat M')\ge \pmarg^{\iA}(M',\hat M')$.
    Hence, the popularity of $M'$ in $\iB$ follows from the popularity of $M'$ in $\iA$.
    This concludes the proof of the first implication.

    $\impliedby$ Conversely, assume that $M'$ is a matching that is popular for both $\iA$ and $\iB$.
    Define the matching $M = \{f\in M'\colon f\subseteq W\cup F\}$.
    We will first show that $M$ covers $d$ and that  $e\notin M$.
    Subsequently, we will conclude by showing that $M$ is popular in the source.

    First, assume for contradiction that $d$ is not covered by $M$.
    Then $d$ is not covered by $M'$ or $\{d,\ell_d\}\in M'$.
    In the former case, consider the matching $\hat M'_1 = \{f\in M'\colon f\subseteq W\cup F\cup \{r_a,\ell_a\}\}\cup \{d,\ell_d\}$.
    Then, all agents $x\in (W\cup F\cup \{r_a,\ell_a\})\setminus\{d\}$ have the identical partners in $M'$ and $\hat M'_1$ (or no partners in both), and, therefore, $\vote_x^{\iA}(\hat M'_1,M') = 0$.
    Moreover, in $\iA$, we have that $\vote_d^{\iA}(\hat M'_1,M') = \vote_{\ell_d}^{\iA}(\hat M'_1,M') = 1$.
    Hence, even if $r_d$ prefers $M'$, it follows that $\pmarg^{\iA}(\hat M'_1,M') > 0$, contradicting that $M'$ is popular in $\iA$.
    Now consider the case when $\{d,\ell_d\}\in M'$.
    Define $\hat M'_2 = \{f\in M'\colon f\subseteq W\cup F\cup \{r_a,\ell_a\}\}\cup \{r_d,\ell_d\}$.
    As in the first case, all agents $x\in (W\cup F\cup \{r_a,\ell_a\})\setminus\{d\}$ satisfy $\vote^{\iB}_x(\hat M'_2,M') = 0$.
    Moreover, we have that $\vote_{r_d}^{\iB}(\hat M'_2,M') = \vote^{\iB}_{\ell_d}(\hat M'_2,M') = 1$.
    Hence, even if $d$ prefers $M'$, it follows that $\pmarg^{\iB}(\hat M'_2,M') > 0$, contradicting that $M'$ is popular in $\iB$.
    We have obtained a contradiction in both cases and conclude that $d$ is covered by $M$.

    Now, assume for contradiction that $e\in M$.
    Then, $e$ must already have been in $M'$.
    Recall that $b$ and $c$ are each other's most preferred partner.
    Hence, if $c$ was unmatched in $M'$, then replacing $e$ by $\{b,c\}$ would be more popular.
    Consider the matching $\hat M'_3 = (M'\setminus \{e,\{c,M(c)\}\})\cup \{\{b,c\},\{a,\ell_a\}\}$.
    We compute the popularity margin between $\hat M'_3$ and $M'$ in $\iB$.
    The only agents that are not necessarily indifferent between $\hat M'_3$ and $M'$ are $r_a$, $\ell_a$, $a$, $b$, $c$, and $M(c)$.
    However, for all $x\in \{\ell_a,a,b,c\}$, it holds that $\vote_{x}^{\iB}(\hat M'_3,M') = 1$.
    Hence, since only $M(c)$ and $r_a$ can be in favor of $M'$, we conclude that $\pmarg^{\iB}(\hat M'_3,M') > 0$, contradicting the popularity of $M'$ in $\iB$.
    Hence, it must hold that $e\notin M$.

    Finally, we claim that $M$ is popular in $\ins$.
    Assume for contradiction that there exists a matching $\hat M$ on $\ins$ with $\pmarg^{\ins}(\hat M, M) > 0$.
    We challenge $M'$ in $\iA$ with a matching that depends on whether or not $a$ is matched in $\hat M$.
    If $a$ is matched in $M$, consider $\hat M'_4 = \hat M \cup \{\{r_a,\ell_a\},\{r_d,\ell_d\}\}$.
    If $a$ is unmatched in $M$, consider $\hat M'_5 = \hat M \cup \{\{a,\ell_a\},\{r_d,\ell_d\}\}$.

    For most agents, we can determine the popularity margins for $\hat M'_4$ and $\hat M'_5$ simultaneously. Therefore, let $i\in \{4,5\}$.
    Now, for all agents in $x\in (W\cup F)\setminus\{a\}$, we have $\hat M(x) = \hat M'_i(x)$ or they are unmatched in both $\hat M$ and $\hat M'_i$ as well as $M(x) = M'(x)$ or they are unmatched in both.
    Hence, $\vote_x^{\iA}(\hat M'_i,M') = \vote_x^{\ins}(\hat M,M)$.
    We refer to this insight as Observation~$(\star)$.

    Moreover, we have that $\vote^{\iA}_{\{r_d,\ell_d\}}(\hat M'_i,M')\ge 0$.
    Indeed, this is immediate if $\{r_d,\ell_d\}\in M'$. Otherwise, we have $\vote^{\iA}_{r_d}(\hat M'_i,M') = 1$ and the inequality holds even if $\ell_d$ is in favor of $M'$.
    We refer to this insight as Observation~$(\star\star)$.

    Next, we know that $\{a,\ell_a\}\in M'$.
    Indeed, since $e\notin M'$, $a$ would be uncovered otherwise.
    Hence, including $\{a,\ell_a\}$ while possibly deleting $\{r_a,\ell_a\}$ would be preferred by $a$ and $\ell_a$ while only $r_a$ can be against it.
    
    We conclude by distinguishing whether $a$ is matched or not in $\hat M$.

    If $a$ is matched in $\hat M$, then, since $b\succ^{\iA}_a \ell_a$, it holds that $\vote_a^{\iA}(\hat M'_4, M') = 1$.
    Moreover, we have that $\vote_{r_a}^{\iA}(\hat M'_4, M') = 1$ while $\vote_{\ell_a}^{\iA}(\hat M'_4, M') = -1$.
    Hence, $\vote_{S_a}^{\iA}(\hat M'_4, M') = 1 = \vote_a^{\ins}(\hat M,M)$.
    Together with Observations~$(\star)$ and $(\star\star)$, we conclude that $\pmarg^{\iA}(\hat M'_4,M') \ge \pmarg^{\ins}(\hat M,M) > 0$, a contradiction. 
    
    Finally, if $a$ is not matched in $\hat M$, then all of $a$, $\ell_a$, and $r_a$ are in the same situation in $M'$ and $\hat M'$ while $a$ is unmatched in both $M$ and $\hat M$. Hence, $\vote_{S_a}^{\iA}(\hat M'_5,M') = 0 = \vote_{a}^{\ins}(M',M')$.
    Again, combining this with Observations~$(\star)$ and $(\star\star)$ yields $\pmarg^{\iA}(\hat M'_5,M') \ge \pmarg^{\ins}(\hat M,M) > 0$, a contradiction. 
    
    Since we have excluded all cases, there cannot be a matching more popular in $M$ in $\ins$.
    This concludes the proof.
\end{proof}

In light of the combination of \Cref{thm:twooneside,thm:hardness}, one can reason about the complexity of \RobustProb (and \RobustDom) when only two agents from the same side perform a swap.
A natural attempt towards a polynomial-time algorithm would be to generalize our technique involving hybrid instances from \Cref{sec:oneagent}.
However, even for this restricted setting, \Cref{lem:hybrid_complete_2}, and, therefore, \Cref{cor:robustedge}, break down.
We provide such an example in \Cref{app:twoagents}. 
We feel that algorithms reach a limitation and conjecture a hardness result even in this case.

\subsection{Unpopular Agents}

We continue the consideration of instances of \RobustProb with a common underlying graph, but from a different angle.
In this section, we consider agents that are not matched by any popular matching.
We refer to such an agent as an \emph{unpopular agent}.
All other agents are called \emph{popular agents}.
Given an instance $\ins$ of {\MatP}, let $\UI$ denote the set of unpopular agents in $\ins$.
The consideration of unpopular agents leads to a class of instances of \RobustProb that are trivially Yes-instances because popular matchings are maintained.

\begin{proposition}\label{prop:unpopular}
    Consider an instance $\iP$ of \RobustProb where only the preference orders of agents in $\UA$, i.e., of unpopular agents in $\iA$, differ in the perturbed instance.
    Then, $\iP$ is a Yes-instance of \RobustProb.
\end{proposition}

\begin{proof}
    Let $\iP$ be an instance of \RobustProb where $\iA$ and $\iB$ only differ with respect to the preference orders of agents in $\UA$.
    Let $M$ be a popular matching in $\iA$.
    We claim that $M$ is also popular for $\iB$.

    Let $M'$ be any other matching.
    Let $x\in (W\cup F)\setminus \UA$ be a popular agent. 
    Then, because the preferences of $x$ are the same in both instances, $\vote_x^{\iB}(M',M) = \vote_x^{\iA}(M',M)$.
    Now, let $x\in \UA$.
    Since $x$ is unmatched in $M$, $x$ votes in favor of $M'$ in both $\iA$ and $\iB$ if $x$ is matched in $M'$ and is indifferent between the two matchings if $x$ remains unmatched.
    Hence, once again $\vote_x^{\iB}(M',M) = \vote_x^{\iA}(M',M)$.
    Together, $\pmarg^{\iB}(M',M) = \pmarg^{\iA}(M',M)$.
    Hence, the popularity of $M$ in $\iB$ follows from the popularity of $M$ in $\iA$.
\end{proof}

Put differently, the computation of robust matchings is not sensitive to perturbances of agents that do not matter to popularity in the first place.
Notably, the set of unpopular agents can be computed efficiently:
We can compute their complement, i.e., the set of popular agents, by simply checking an instance of \PopEdge for every available edge.
Moreover, like for perturbations of one agent in \Cref{thm:multiple_instances}, \Cref{prop:unpopular} extends to multiple instances if these only differ with respect to perturbances of the preferences of the unpopular agents in one of these instances.

\subsection{Reduced Availability}\label{sec:redavb}

We turn to the consideration of \RobustProb for the case of alternated availability, i.e., the underlying graph may change while maintaining preference orders among common edges.
In particular, we consider the special case where the underlying graph is complete.
To this end, an instance $\ins$ of {\MatP} is said to be \emph{complete} if $G^{\ins}$ is the complete bipartite graph on vertex set $W\cup F$, i.e., the edge set is the Cartesian product of the set of workers and firms $E^{\ins} = W \times F$.
Note that if one of the {\MatP} instances of a \RobustProb instance is complete, then alternated availability is identical to reduced availability.
Our first result is an efficient algorithm for this case.

\begin{proposition}\label{prop:ReduceAcc}
    \RobustProb can be solved in polynomial time for input instances $\iP$ where $\iB$ evolves from $\iA$ by reducing availability and $\iA$ is complete.
\end{proposition}

\begin{proof}
    We show how to solve the problem by solving a maximum weight popular matching problem.
    Consider an instance $\ins$ for {\MatP} and assume that we are given a weight function $w\colon E^{\ins}\to \mathbb Q$.
    The \emph{weight} of a matching is defined as $w(M) :=\sum_{e\in M} w(e)$.
    It is known that the problem of computing a matching of maximum weight among popular matchings can be solved in polynomial time for complete instances of {\MatP} \citep{CsKa17a}.

    Now, consider an instance $\iP$ of \RobustProb where $\iB$ evolves from $\iA$ by reducing availability and $\iA$ is complete.
    We define the weight function $w\colon E^{\iA} \to \{-1,0\}$ by $w(e) = 0$ if $e\in E^{\iB}$ and $w(e) = -1$, otherwise. 
    
    We claim that $\iP$ is a Yes-instance of \RobustProb if and only if the maximum weight popular matching in $\iA$ with respect to $w$ has a weight of~$0$.

    First, assume that $M$ is a popular matching for both $\iA$ and $\iB$.
    Then, $M\subseteq E^{\iB}$ and $M$ is a popular matching for $\iA$ with $w(M) = 0$.

    Conversely, if $M$ is a popular matching for $\iA$ with $w(M) = 0$.
    Then, $M\subseteq E^{\iB}$.
    Moreover, any other matching for $\iB$ is also a matching in $\iA$ with the identical popularity margin.
    Hence, the popularity of $M$ for $\iB$ follows from the popularity of $M$ for $\iA$.
\end{proof}

Interestingly, we can still extend \Cref{prop:ReduceAcc} to the case of multiple instances.
As long as one instance is complete, all other instances may differ by arbitrary altered availability.
For a proof, we can simply adjust the weight function in the proof of \Cref{prop:ReduceAcc} to be $0$ for edges present in \emph{all} instances.

\begin{proposition}
    There exists a polynomial-time algorithm for the following problem:
    Given a collection of {\MatP} instances $(\ins_1,\dots, \ins_k)$, where $\ins_1$ is complete and all instances only differ by altered availability, does there exist a matching that is popular for $\ins_i$ for all $1\le i\le k$?
\end{proposition}

\begin{proof}
    Assume that $(\ins_1,\dots, \ins_k)$ is a collection of {\MatP} instances, where $G^{\ins_i} = (W\cup F,E_i)$ for $1\le i \le k$ and $E_1 = W\times F$, i.e., the first instance $\ins_1$ is complete.
    We claim that there exists a matching that is popular for $\ins_i$ for all $1\le i\le k$ if and only if there exists a popular matching $M$ for $\ins_1$ with $M\subseteq \bigcap_{i=1}^k E_i$.
    First, if $M$ is popular for $\ins_i$ for all $1\le i\le k$, then $M$ is in particular popular for $\ins_1$.
    Moreover, since $M$ is a matching for each instance, it holds that $M\subseteq E_i$ for all $1\le i\le k$.
    
    Conversely, assume that there exists a popular matching $M$ for $\ins_1$ with $M\subseteq \bigcap_{i=1}^k E_i$.
    Let $2\le i\le k$.
    Then, $M$ is a matching in $\ins_i$.
    Moreover, any other matching for $\ins_i$ is also a matching in $\ins_1$ with the identical popularity margin.
    Hence, the popularity of $M$ for $\ins_i$ follows from the popularity of $M$ for $\ins_1$.

    Now, consider the weight function $w\colon E^{\ins_1} \to \{-1,0\}$ by $w(e) = 0$ if $e\in \bigcup_{i=1}^k E_i$ and $w(e) = -1$, otherwise.
    Then, $M$ is a popular matching $M$ for $\ins_1$ with $M\subseteq \bigcap_{i=1}^k E_i$ if and only if the maximum weight popular matching in $\ins_1$ with respect to $w$ has weight $0$.
    Hence, as in \Cref{prop:ReduceAcc}, we can compute and check such a matching in polynomial time.
\end{proof}

However, the restriction that $\iA$ is a complete instance is essential for \Cref{prop:ReduceAcc} to hold.
If we drop it, we obtain a computational intractability as we show next.
The proof of this result is a simpler version of the proof of \Cref{thm:hardness} that works without auxiliary agents.

\begin{restatable}{proposition}{PropNPRedAva}
    \RobustProb is \NP-complete for input instances $\iP$ where $\iB$ evolves from $\iA$ by reducing availability.
\end{restatable}

\begin{proof}
    Membership in {\NP} holds as in the proof of \Cref{thm:hardness}.
    For \NP-hardness, we perform another reduction from \ForbEdge as defined in \Cref{sec:problems}.
    Consider an instance of \ForbEdge given by an instance $\ins$ of \MatP on the graph $(W\cup F, E)$ together with two designated edges $e,e'\in E$. 
    We define a reduced instance $\iP$ of \RobustProb where $\iA$ is identical to $\ins$.
    Moreover, $\iB$ is defined by $G^{\iB} = (W\cup F, E^{\ins}\setminus \{e,e'\})$.
    For agent $x\in W\cup F$, preferences are inherited from $\ins$, i.e., for all $y,z\in N_x^{\iB}$, we define $y\succ_x^{\iB} z$ if and only if $y\succ_x^{\ins} z$.

    If $M$ is a popular matching for $\ins$ with $M\cap \{e,e'\} = \emptyset$, then $M$ is popular for $\iA$ because $\iA$ is identical to $\ins$.
    It is also popular for $\iB$ because any other matching for $\iB$ is also a matching in $\ins$ with the identical popularity margin.
    Conversely, if $M$ is a robust popular matching for $\iP$, then $M\cap \{e,e'\} = \emptyset$ because $M$ is a matching for $\iB$.
    In addition, $M$ is popular for $\ins$ because it is popular for $\iA$.
    This establishes the correctness of the reduction and completes the proof.
\end{proof}

Notably, by inspection of the proof, \RobustProb is already \NP-complete if the input instances only differ by reducing availability where two edges are removed.
By contrast, if only one edge is removed, the problem is equivalent to computing a popular matching with a single forbidden edge.
This problem can be solved in polynomial time \citep{FKPZ19a}.

\subsection{Robust Popular Matchings in Related Models}
We conclude our result section by discussing robustness of popularity in related models. 

First, robustness of matchings can be defined for other models of popularity.
As mentioned earlier, there exists the concept of strongly popular matchings, which have a strictly positive popularity margin against any other matching.
Since strongly popular matchings are unique and can be computed in polynomial time \citep{BIM10a}, robust strongly popular matchings can also be computed in polynomial time, whenever they exist: One can simply check if strongly popular matchings exist in all input instances and compare them.

Second, one can consider popularity for mixed matchings, which are probability distributions over deterministic matchings, and popularity is then defined as popularity in expectation \citep{KMN11a}.
Popular mixed matchings correspond to the points of a tractable polytope for which feasible points can be identified in polynomial time.
One can intersect the polytopes for multiple instances and still obtain a tractable polytope.
This approach yields a polynomial time algorithm to solve \RobustProb for mixed matchings and can even be applied for roommate games.
In these games, the input graph is not required to be bipartite anymore and the linear programming method can still be applied \citep{BrBu20a}.
Notably, this approach cannot be used to determine deterministic matchings.
Even if the polytopes for all input instances are integral, the intersection of the polytopes may be nonempty but not contain integral points anymore.
We discuss technical details concerning mixed popularity including
such an example in \Cref{app:LPlimit}.

Finally, one can consider popularity for more general input instances.
However, this quickly leads to intractabilities because the existence of popular matchings may not be guaranteed any more.
For example, it is \NP-hard to decide whether a popular matching exists if we consider roommate games \citep{FKPZ19a} or if we have bipartite graphs but weak preferences \citep{BIM10a}.
These results immediately imply \NP-hardness of \RobustProb because one can simply duplicate the source instance, and a robust popular matching exists if and only if the source instance admits a popular matching.

\section{Conclusion}
We have initiated the study of robustness for popular matchings by considering the algorithmic question of determining a popular matching across two or more instances of matching under preferences.
We investigated this problem for two restrictions.
First, we assumed that agents perturb only their preferences over a fixed set of available matching partners.
When a single agent perturbs their preference order, we presented a polynomial-time algorithm for solving \RobustProb, based on solving \PopEdge on suitably defined hybrid instances.
We demonstrate the applicability of this approach by using it to solve \RobustDom as well.
By contrast, we showed that \RobustProb and \RobustDom become \NP-complete even for the case where only two agents from the same side perturb their preferences. 
We obtain another hardness result for \RobustProb if only two (different-side) agents perform the simple operation of a swap of two adjacent alternatives.
Additionally, we identified a class of Yes-instances for \RobustProb in which only unpopular agents perturb their preference orders.

Moreover, we considered \RobustProb under reduced availability.
We established a complexity dichotomy based on preference completeness.
If one input instance is complete, \RobustProb can be solved efficiently by reducing it to a maximum-weight popular matching problem.
However, if this is not the case, we once again obtain an \NP-completeness result.

We believe that our work offers various exciting future research directions, and we conclude by discussing some of these.
First, as discussed at the end of \Cref{sec:2hard}, an immediate open problem is the complexity of \RobustProb and \RobustDom if only two agents from one side may change their preference orders by performing a swap. 
Notably, our approach of defining hybrid instances already faces limitations in this restricted case.
We, therefore, conjecture hardness even for this case.
Another specific open problem is to consider \RobustDom for reduced availability, since we only considered \RobustProb in \Cref{sec:redavb}.

A further interesting avenue would be to consider the problem of computing a matching of maximum size among robust popular matchings, whenever such a matching exists.
Note that, since dominant matchings are maximum-size popular matchings, this task can be performed in polynomial time for a single instance \citep{CsKa17a}.
This problem could, however, be more involved than robust popularity or robust dominance.
In fact, a maximum-size robust popular matching is not the same as a robust dominant matching, even if only one agent perturbs their preferences, see \Cref{app:dominant}.
We suspect hardness for computing maximum-size robust popular matchings even for this case. 

On a different note, it would be interesting to explore escape routes to our discovered hardness results.
For this, one could try to efficiently find matchings offering a compromise between popularity in each of the input instances.
For instance, one could attempt to find popular matchings in the second instance that have a large overlap with a given popular matching in the first instance.
For complete instances, this can be done by finding a maximum weight popular matching problem similar to the approach for \Cref{prop:ReduceAcc}.
In general, defining and investigating other notions of compromise matchings may lead to intriguing further discoveries.

\section*{Acknowledgements}
    A preliminary version of this paper appeared in the proceedings of the 23rd International Conference on Autonomous Agents and Multiagent Systems
    (AAMAS 2024).
    Most of this work was done when Martin Bullinger was at the University of Oxford.
	This work was supported in part by the AI Programme of The Alan Turing Institute and the NSF under grant CCF-2230414.
    We would like to thank Vijay Vazirani and Telikepalli Kavitha for their valuable discussions, and the anonymous reviewers from AAMAS for their comments.

\clearpage

\appendix

\section*{Appendix}

In the appendix, we provide further insight on challenges for applying our technique based on hybrid instances for two agents permuting their preference orders, as well as limitations of approaches based on linear programming.

\section{Hybrid Instances for Two Agents}\label{app:twoagents}

In this appendix, we discuss difficulties for defining hybrid instances similar to \Cref{sec:oneagent} if two agents perturb their preferences. 
The natural generalization of our approach using hybrid instances for a number of agents changing their preferences larger than one would be to 
\begin{enumerate}
    \item consider any possible combination of matching partners for the agents changing preference orders,
    \item for each combination, define a hybrid instance by moving all agents preferred to the designated partners in any input instance to the top of their preference orders, and
    \item investigate whether the hybrid instance contains a popular matching containing all of the designated edges.
\end{enumerate}

Unfortunately, this road map leads to multiple challenges.
The first step already has an exponential blowup with respect to the number of agents changing their preferences.
While this is undesirable, we might however still find a reasonable approach for a small fixed number of agents or even a fixed-parameter tractability with respect to the number of agents changing their preferences.

Moreover, the last step of this construction can lead to computational problems because finding a popular matching with at least two forced edges is \NP-complete \citep{FKPZ19a}.
As a solution, we could further restrict our input.
If we assume that the input instances are complete, we can solve the algorithm problem of finding a popular matching containing any subset of edges in polynomial time by solving a maximum weight popular matching problem 
\citep{CsKa17a}.
So, let us assume that we can deal with this challenge.

However, even then, we face further difficulties.
The following example shows that \Cref{lem:hybrid_complete_2} breaks down.

\begin{figure*}
	\begin{wbox}
	\centering
	\begin{minipage}{.25\textwidth}
        \centering
		\resizebox{.98\textwidth}{!}{
		\begin{tabular}{l|llll}
			$w_1$  & \color{myred}$f_1$ & \color{myred}$f_2$ & $f_3$ & $f_4$\\
			$w_2$  & \color{myred}$f_3$ & \color{myred}$f_4$ & $f_1$ & $f_2$\\
			$w_3$  & $f_4$ & $f_2$ & $f_1$ & $f_3$\\
			$w_4$  & $f_3$ & $f_1$ & $f_2$ & $f_4$
        \end{tabular} }
  
		\hspace{1cm}
		
		{\small Workers' preferences\\ in $\iA$}
		
		\hspace{1cm}
		
	\end{minipage}%
	\begin{minipage}{.25\textwidth}
        \centering
		\resizebox{.98\textwidth}{!}{
		\begin{tabular}{l|llll}
			$w_1$  & \color{myred}$f_2$ & \color{myred}$f_1$ & $f_3$ & $f_4$\\
			$w_2$  & \color{myred}$f_4$ & \color{myred}$f_3$ & $f_1$ & $f_2$\\
			$w_3$  & $f_4$ & $f_2$ & $f_1$ & $f_3$\\
			$w_4$  & $f_3$ & $f_1$ & $f_2$ & $f_4$
        \end{tabular}}
  
		\hspace{1cm}
		
		{\small Workers' preferences\\ in $\iB$} 
		
		\hspace{1cm}
		
	\end{minipage}%
	\begin{minipage}{.25\textwidth}
        \centering
		\resizebox{.98\textwidth}{!}{
		\begin{tabular}{l|llll}
			$f_1$  & $w_1$ & $w_2$ & $w_3$ & $w_4$\\
			$f_2$  & $w_1$ & $w_2$ & $w_3$ & $w_4$\\
			$f_3$  & $w_2$ & $w_4$ & $w_1$ & $w_3$\\
			$f_4$  & $w_2$ & $w_4$ & $w_1$ & $w_3$
        \end{tabular}}
  
		\hspace{1cm}
		
		{\small Firms' preferences\\ in $\iA$ and $\iB$}
		
		\hspace{1cm}
		
	\end{minipage}%
	\begin{minipage}{.25\textwidth}
        \centering
		\resizebox{.98\textwidth}{!}{
		\begin{tabular}{l|llll}
			$w_1$  & \color{myred}$f_2$ & \color{myred}$f_1$ & $f_3$ & $f_4$\\
			$w_2$  & \color{myred}$f_3$ & \color{myred}$f_4$ & $f_1$ & $f_2$\\
			$w_3$  & $f_4$ & $f_2$ & $f_1$ & $f_3$\\
			$w_4$  & $f_3$ & $f_1$ & $f_2$ & $f_4$
        \end{tabular}}
  
		\hspace{1cm}
		
		{\small Workers' preferences\\ in $\iHH$} 
		
		\hspace{1cm}
		
	\end{minipage}%
	\end{wbox}
	\caption{Difficulties for defining hybrid instances when two agents swap the preference order for two possible matching partners.
    In \Cref{ex:twoswap}, we define an instance $\iP$ where $w_1$ and $w_2$ both swap the preference order for their two most preferred firms.
    The matching $M = \{\{w_1,f_1\}, \{w_2,f_4\}, \{w_3,f_2\}, \{w_4,f_3\}\}$ is popular for both $\iA$ and $\iB$ but not for $\iHH$.}
	\label{fig:twoswap} 
\end{figure*}

\begin{example}\label{ex:twoswap}
Consider the instance $\iP$ defined in \Cref{fig:twoswap}.
In both instances of \MatP, the set of workers and firms is $W = \{w_1, w_2, w_3, w_4\}$ and $F = \{f_1, f_2, f_3, f_4\}$.
The firms have identical preference orders in both instances, whereas the workers $w_1$ and $w_2$ both swap the preference order for their two most preferred firms.
We want to check whether there exists a common popular matching containing $e = \{w_1,f_1\}$ and $e' = \{w_2, f_4\}$.
To this end, we define a hybrid instance $\iHH$, where $w_1$ ranks all firms preferred to $f_1$ in $\iA$ and $\iB$ above $f_1$, and $w_2$ ranks all firms preferred to $f_4$ in $\iA$ and $\iB$ above $f_4$.
We leave the remaining preference orders as in $\iA$ or in $\iB$.
The resulting preference orders are depicted in \Cref{fig:twoswap} to the very right.

Consider the matching $$M = \{\{w_1,f_1\}, \{w_2,f_4\}, \{w_3,f_2\}, \{w_4,f_3\}\}\text.$$
One can show that $M$ is popular for both $\iA$ and $\iB$.
A fast way to do this is to use the combinatorial description of popular matchings of \Cref{thm:pop_char}.
We draw a graph with vertex set $W\cup F$ and edges corresponding to the matching $M$.
Then, we add the $(+,+)$ and $(+,-)$ edges for agent pairs such that both or one of the edge's endpoints prefer the partner in this edge to their matching partner in $M$.
We draw a straight red edge for a $(+,+)$ edge that is preferred by both endpoints and a dashed red edge for a $(+,-)$ that is preferred by precisely one of its endpoints.
The left graph in \Cref{fig:TwoSwapPop} depicts this situation for $M$ in $\iA$ and the right graph for $M$ in $\iB$.
According to the characterization in \Cref{thm:pop_char}, $M$ is popular in both instances because it admits no undesired alternating paths or cycles.

\begin{figure}
    \centering
    \begin{tikzpicture}
            \pgfmathsetmacro\graphspan{2.5}
            \pgfmathsetmacro\graphheight{1.2}
            \pgfmathsetmacro\nodesize{.6cm}
            \node[draw, circle, minimum size = \nodesize] (w1) at (0,\graphheight) {$w_1$};
            \node[draw, circle, minimum size = \nodesize] (w2) at (0,0) {$w_2$};
            \node[draw, circle, minimum size = \nodesize] (w3) at (0,-\graphheight) {$w_3$};
            \node[draw, circle, minimum size = \nodesize] (w4) at (0,-2*\graphheight) {$w_4$};

            \node[draw, circle, minimum size = \nodesize] (f1) at (\graphspan,\graphheight) {$f_1$};
            \node[draw, circle, minimum size = \nodesize] (f2) at (\graphspan,0) {$f_2$};
            \node[draw, circle, minimum size = \nodesize] (f3) at (\graphspan,-\graphheight) {$f_3$};
            \node[draw, circle, minimum size = \nodesize] (f4) at (\graphspan,-2*\graphheight) {$f_4$};

            \node at (.5*\graphspan, -2.7*\graphheight) {$\iA$};
            
            \draw[thick] (w1) -- (f1);
            \draw[thick] (w2) -- (f4);
            \draw[thick] (w3) -- (f2);
            \draw[thick] (w4) -- (f3);

            \draw[thick, red, dashed] (w3) -- (f4);
            \draw[thick, red, dashed] (w1) -- (f2);
            \draw[thick, red] (w2) -- (f3);
    \end{tikzpicture}
    \qquad
    \begin{tikzpicture}
            \pgfmathsetmacro\graphspan{2.5}
            \pgfmathsetmacro\graphheight{1.2}
            \pgfmathsetmacro\nodesize{.6cm}
            \node[draw, circle, minimum size = \nodesize] (w1) at (0,\graphheight) {$w_1$};
            \node[draw, circle, minimum size = \nodesize] (w2) at (0,0) {$w_2$};
            \node[draw, circle, minimum size = \nodesize] (w3) at (0,-\graphheight) {$w_3$};
            \node[draw, circle, minimum size = \nodesize] (w4) at (0,-2*\graphheight) {$w_4$};

            \node[draw, circle, minimum size = \nodesize] (f1) at (\graphspan,\graphheight) {$f_1$};
            \node[draw, circle, minimum size = \nodesize] (f2) at (\graphspan,0) {$f_2$};
            \node[draw, circle, minimum size = \nodesize] (f3) at (\graphspan,-\graphheight) {$f_3$};
            \node[draw, circle, minimum size = \nodesize] (f4) at (\graphspan,-2*\graphheight) {$f_4$};

            \node at (.5*\graphspan, -2.7*\graphheight) {$\iB$};

            \draw[thick] (w1) -- (f1);
            \draw[thick] (w2) -- (f4);
            \draw[thick] (w3) -- (f2);
            \draw[thick] (w4) -- (f3);

            \draw[thick, red, dashed] (w3) -- (f4);
            \draw[thick, red] (w1) -- (f2);
            \draw[thick, red, dashed] (w2) -- (f3);
    \end{tikzpicture}
    \caption{Certificate of popularity of $M$ for both $\iA$ and $\iB$.
    A straight red and dashed red edges correspond to $(+,+)$ and $(+,-)$ edges, respectively.}
    \label{fig:TwoSwapPop}
\end{figure}

On the other hand, $M$ is not popular for $\iHH$.
Indeed, consider the matching $$M' = \{\{w_1,f_2\}, \{w_2,f_3\}, \{w_3,f_4\}, \{w_4,f_1\}\}\text.$$
Then, $M'$ is preferred by $w_1, w_2, w_3, f_2$, and $f_3$, whereas $M$ is preferred by $w_4, f_1$, and $f_4$.
Hence, $\pmarg^{\iHH}(M',M) = 2 > 0$ and $M$ is not popular.

Together, this example shows that \Cref{lem:hybrid_complete_2} does not generalize for the extension of hybrid instances considered here.
\hfill$\lhd$
\end{example}

In contrast to the previous example, one can show that a variant of \Cref{lem:hybrid_complete_1} still holds if we generalize hybrid instances similar to $\iHH$.
Nonetheless, we have seen in this section that our approach with hybrid instances causes new challenges for solving \RobustProb if two agents change their preferences, even if these agents are from the same side and only perform a swap of two agents.

\section{Limitations of LP-based Approaches}\label{app:LPlimit}

A common approach to tackle algorithmic questions concerning popularity is to apply linear programming \citep[see, e.g.,][]{KMN11a,BrBu20a}.
Indeed, one can define the so-called popularity polytope that contains all fractional popular matchings. 
These can be interpreted as probability distributions over deterministic matchings which have a non-negative popularity margin against every other matching in expectation.

Let us provide the formal framework for this approach.
Assume that we are given a complete instance $\ins$ of {\MatP} based on a graph $G = (W\cup F, E)$.
We can consider its popularity polytope  $\mathcal P \subseteq [0,1]^E$ defined as
\begin{align*}
\mathcal P := \{\mu\in \mathbb R ^ E \colon 
\sum_{e\in E, x\in e} \mu(e) &= 1\ \forall x\in W\cup F,\\
 \pmarg^{\ins}(\mu,\chi_M) &\ge 0\ \forall \text{ matchings } M\\
\mu(e) &\ge 0\ \forall e\in E\}\text.	
\end{align*}

There, $\chi_M$ denotes the incidence vector of a matching $M$.
The first set of constraints and the nonnegativity constraints define the \emph{perfect matching polytope} for this graph \citep{Edmo65b}.
The important set of constraints are given by the inequalities $\pmarg^{\ins}(\mu,\chi_M) \ge 0$, which denotes the expected popularity margin against an incident vector.\footnote{The precise definition for the generalized popularity margin for fractional matchings can be found in the paper by \citet{KMN11a}.}
Importantly, the integral points of $\mathcal P$ correspond to popular matchings of $\ins$.
Feasible points in the polytope can be found by solving the separation problem using an algorithm proposed by \citet{McCu08a}.
Moreover, the polytope is integral if it is based on an instance of {\MatP}, where the stable matchings are perfect matchings.
This is in particular the case for complete input instances.
In general, however, this polytope is only half-integral \citep{HuKa21a}.

We want to discuss prospects of applying the popularity polytope to aid solving \RobustProb.
The natural try would be to consider the combination of the constraints of the popularity polytopes of $\iA$ and $\iB$, provided that we are given an instance $\iP$ of \RobustProb.
It is then possible to check whether this polytope is nonempty, and this actually solves \RobustProb for mixed matchings.

However, the resulting polytope has limitations for answering questions about deterministic robust popular matchings.
Indeed, it can happen that it contains feasible solutions but no integral solution.
Hence, in particular, the resulting polytope need not be integral even if both the polytopes corresponding to the input instances are integral.

\begin{example}\label{ex:LPbound}
    Consider an instance $\iP$ of \RobustProb where $W = \{w_1, w_2, w_3\}$ and $F = \{f_1, f_2,, f_3\}$.
    The preference orders of the agents are as defined in \Cref{fig:LPboundary}.

\begin{figure*}
	\begin{wbox}
	\centering
	\begin{minipage}{.33\linewidth}
        \centering
		\begin{tabular}{l|lll}
			$w_1$  & $f_3$ & $f_1$ & $f_2$\\
			$w_2$  & $f_3$ & $f_1$ & $f_2$\\
			$w_3$  & $f_3$ & $f_1$ & $f_2$
        \end{tabular}
  
		\hspace{1cm}
		
		Workers' preferences in $\iA$ 
		
		\hspace{1cm}
		
	\end{minipage}%
	\begin{minipage}{.33\linewidth}
        \centering
		\begin{tabular}{l|lll}
			$w_1$  & $f_3$ & $f_2$ & $f_1$\\
			$w_2$  & $f_3$ & $f_2$ & $f_1$\\
			$w_3$  & $f_3$ & $f_2$ & $f_1$
        \end{tabular}
 
		\hspace{1cm}
		
		Workers' preferences in $\iB$ 
		
		\hspace{1cm}
		
	\end{minipage}%
	\begin{minipage}{.33\linewidth}
        \centering
		\begin{tabular}{l|lll}
			$f_1$  & $w_1$ & $w_2$ & $w_3$\\
			$f_2$  & $w_1$ & $w_2$ & $w_3$\\
			$f_3$  & $w_1$ & $w_2$ & $w_3$\\
        \end{tabular}
  
		\hspace{1cm}
		
		Firms' preferences in both
		
		\hspace{1cm}
		
	\end{minipage}%
	\end{wbox}
	\caption{Boundaries of linear programming approaches based on the popularity polytope in \Cref{ex:LPbound}.}
	\label{fig:LPboundary} 
\end{figure*} 

    It can be shown that the popular matchings in $\iA$ are given by
    \begin{itemize}
        \item $M^1_A = \{\{w_1,f_1\},\{w_2,f_2\},\{w_3,f_3\}\}$,
        \item $M^2_A = \{\{w_1,f_3\},\{w_2,f_1\},\{w_3,f_2\}\}$, and
        \item $M^3_A = \{\{w_1,f_2\},\{w_2,f_3\},\{w_3,f_1\}\}$.
    \end{itemize}
    
    In addition, the popular matchings in $\iB$ are given by
    \begin{itemize}
        \item $M^1_B = \{\{w_1,f_1\},\{w_2,f_3\},\{w_3,f_2\}\}$,
        \item $M^2_B = \{\{w_1,f_2\},\{w_2,f_1\},\{w_3,f_3\}\}$, and
        \item $M^3_B = \{\{w_1,f_3\},\{w_2,f_2\},\{w_3,f_1\}\}$.
    \end{itemize}
    
    As a consequence, $\iP$ is a No-instance of \RobustProb.
    However, the joint polytope for both $\iA$ and $\iB$ contains exactly the fractional matching that puts probability $\frac 13$ on any possible edge.
    This corresponds to putting probability $\frac 13$ on each of the popular matchings in one of the instances $\iA$ or~$\iB$.
    Hence, the joint polytope is nonempty but contains no integral point and, therefore, no point corresponding to a deterministic matching.
    \hfill$\lhd$
\end{example}

\section{Maximum-Size Robust Popular Matchings}
\label{app:dominant}

In this section, we want to provide an example to show that computing a robust popular matching of maximum size is not the same as computing a robust dominant matching, i.e., a matching that is a maximum-size popular matching in all involved instances.

Clearly, a robust dominant matching is a maximum-size robust popular matching.
However, as our final example shows, the converse is not necessarily true.

\begin{figure*}
	\begin{wbox}
	\centering
	\begin{minipage}{.33\linewidth}
        \centering
		\begin{tabular}{l|lll}
			$w_1$  & $f_1$ & $f_3$\\
			$w_2$  & $f_1$ & $f_2$\\
			$w_3$  & $f_3$
        \end{tabular}
  
		\hspace{1cm}
		
		Workers' preferences in $\iA$ 
		
		\hspace{1cm}
		
	\end{minipage}%
	\begin{minipage}{.33\linewidth}
        \centering
		\begin{tabular}{l|lll}
			$w_1$  & $f_3$ & $f_1$\\
			$w_2$  & $f_1$ & $f_2$\\
			$w_3$  & $f_3$
        \end{tabular}
 
		\hspace{1cm}
		
		Workers' preferences in $\iB$ 
		
		\hspace{1cm}
		
	\end{minipage}%
	\begin{minipage}{.33\linewidth}
        \centering
		\begin{tabular}{l|lll}
			$f_1$  & $w_2$ & $w_1$\\
			$f_2$  & $w_2$\\
			$f_3$  & $w_1$ & $w_3$\\
        \end{tabular}
  
		\hspace{1cm}
		
		Firms' preferences in both
		
		\hspace{1cm}
		
	\end{minipage}%
	\end{wbox}
	\caption{Instance pair in \Cref{ex:maxsize} where a maximum-size robust popular matching is not a robust dominant matching.
	\label{fig:maxsize}}
\end{figure*} 

\begin{example}\label{ex:maxsize}
    Consider an instance pair $\iP$ of two \MatP instances where $W = \{w_1, w_2, w_3\}$ and $F = \{f_1, f_2,, f_3\}$.
    The preference orders of the agents are as defined in \Cref{fig:maxsize}.
    Note that the two instances only differ by a swap in the preferences of $w_1$.

    In $\iA$, the matching $M_1 = \{\{w_1,f_3\},\{w_2,f_1\}\}$ is stable and, therefore, popular, whereas $M_2 = \{\{w_1,f_1\},\{w_2,f_2\},\{w_3,f_3\}$ is another popular matching.
    However, in $\iB$, while $M_1$ is still a stable and popular matching, there are no other popular matchings.
    In particular, $M_2$ is less popular than $M_1$ which is preferred by $w_1$, $w_2$, $f_1$, and $f_3$.

    Hence, the instance does not admit a robust dominant matching.
    However, $M_1$ is the unique robust popular matching.
    Therefore, $M_1$ is the maximum-size robust popular matching, but not a robust dominant matching.\hfill$\lhd$
\end{example}

\end{document}